\documentclass[sigconf]{aamas}

\usepackage{balance} 

\usepackage{algorithm}
\usepackage[noend]{algpseudocode}
\usepackage{bm}
\usepackage{enumitem}
\usepackage{multirow}
\usepackage{mathtools}
\usepackage{nicefrac}
\usepackage{longtable,booktabs,geometry}

\newcommand{\G}{\mathcal{G}}
\newcommand{\mat}[1]{\boldsymbol{#1}}

\newcommand{\R}{\mathbb{R}}

\newcommand{\vect}[1]{\boldsymbol{#1}}

\algrenewcommand\algorithmicrequire{\textbf{Input:}}
\algrenewcommand\algorithmicensure{\textbf{Output:}}

\makeatletter
\gdef\@copyrightpermission{
  \begin{minipage}{0.2\columnwidth}
   \href{https://creativecommons.org/licenses/by/4.0/}{\includegraphics[width=0.90\textwidth]{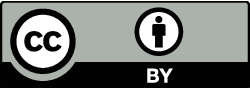}}
  \end{minipage}\hfill
  \begin{minipage}{0.8\columnwidth}
   \href{https://creativecommons.org/licenses/by/4.0/}{This work is licensed under a Creative Commons Attribution International 4.0 License.}
  \end{minipage}
  \vspace{5pt}
}
\makeatother

\setcopyright{ifaamas}
\acmConference[AAMAS '25]{Proc.\@ of the 24th International Conference
on Autonomous Agents and Multiagent Systems (AAMAS 2025)}{May 19 -- 23, 2025}
{Detroit, Michigan, USA}{Y.~Vorobeychik, S.~Das, A.~Nowé  (eds.)}
\copyrightyear{2025}
\acmYear{2025}
\acmDOI{}
\acmPrice{}
\acmISBN{}

\acmSubmissionID{683}

\title{Data Pricing for Graph Neural Networks without Pre-purchased Inspection}

\author{Yiping Liu}
\affiliation{UESTC \country{China}}
\affiliation{The University of Auckland \country{New Zealand}}
\email{yliu823@aucklanduni.ac.nz}
\orcid{0000-0002-0868-6946}
\authornote{Equal contribution}

\author{Mengxiao Zhang}
\affiliation{UESTC \country{China}}
\affiliation{The University of Auckland \country{New Zealand}}
\email{mengxiao.zhang@auckland.ac.nz}
\orcid{0000-0002-6274-0384}
\authornote{Corresponding authors}
\authornotemark[1]

\author{Jiamou Liu}
\affiliation{
The University of Auckland
 \country{New Zealand}}
 \orcid{0000-0002-0824-0899}
\email{jiamou.liu@auckland.ac.nz}

\author{Song Yang}
\affiliation{The University of Auckland
 \country{New Zealand}}
\email{syan382@aucklanduni.ac.nz}
\orcid{0000-0002-1200-1129}

\begin{abstract}
Machine learning (ML) models have become essential tools in various scenarios. Their effectiveness, however, hinges on a substantial volume of data for satisfactory performance. Model marketplaces have thus emerged as crucial platforms bridging model consumers seeking ML solutions and data owners possessing valuable data. These marketplaces leverage model trading mechanisms to properly incentive data owners to contribute their data, and return a well performing ML model to the model consumers. However, existing model trading mechanisms often assume the data owners are willing to share their data before being paid, which is not reasonable in real world. Given that, we propose a novel mechanism, named Structural Importance based Model Trading (SIMT) mechanism, that assesses the data importance and compensates data owners accordingly without disclosing the data. Specifically, SIMT procures feature and label data from data owners according to their structural importance, and then trains a graph neural network for model consumers. Theoretically, SIMT ensures incentive compatible, individual rational and budget feasible. The experiments on five popular datasets validate that SIMT consistently outperforms vanilla baselines by up to $40\%$ in both MacroF1 and MicroF1.
\end{abstract}


\newcommand{\BibTeX}{\rm B\kern-.05em{\sc i\kern-.025em b}\kern-.08em\TeX}

\begin{document}


\pagestyle{fancy}
\fancyhead{}


\maketitle 


\section{Introduction}
In today’s digital age, data has become an essential asset, serving as the foundation for AI and machine learning advancements. To meet the increasing demand for high-quality data, a new business paradigm known as the {\em model marketplace} has emerged \cite{liu2020dealer}, exemplified by platforms like Modzy. A model marketplace facilitates the exchange between {\em model consumers}, who seek AI models for various tasks, and {\em data owners}, who possess the feature and label data necessary for model training. The marketplace purchases data from data owners, uses it to train AI models, and then sells these trained models to consumers. 
However, a key challenge in model marketplaces is determining how to properly compensate data owners for their contributions, a problem referred to as {\em data pricing}. This problem is challenging because the importance of data is difficult to evaluate. Most existing studies assume that marketplaces acquire data from data owners {\em before} paying them and use the subsequent performance improvements as a measure of data importance. For example, \cite{xu2021gradient,agarwal2019marketplace,ghorbani2019data,jia2019efficient,liu2020dealer} rely on this assumption to establish pricing mechanisms based on the marginal impact of data on model accuracy. However, this pre-purchased inspection assumption is impractical in real-world settings. Data owners are often unwilling to release their data without proper payment, fearing that the data, once disclosed, may immediately provide valuable insights to buyers, reducing the incentive to pay.

This leads to a critical question: {\em How can we measure data importance for model training without direct inspection, thereby facilitating data pricing?} Several studies have attempted to address this question by introducing exogenous metrics for measuring data importance, such as data age, accuracy, volume \cite{heckman2015pricing}, the extent of perturbations \cite{cong2020vcg,sun2022profit}, or data owners' reputation \cite{zhang2021incentive}. However, these metrics often fail to accurately reflect the contribution of data in the context of model training, particularly when dealing with complex models like Graph Neural Networks (GNNs).

Graph-structured data is prevalent in many real-world scenarios, where the relationships between entities are often as important as their attributes.
GNNs excel in tasks involving such data, capturing both node features and network structure. 
However, data ownership is often decentralized, with different entities controlling separate “pockets” of the network.  
This creates a need for a marketplace where subgraphs can be purchased and integrated to enable comprehensive model training \cite{bechtsis2022data}. 
For example, in finance, each bank holds its own subset of transaction data, but detecting fraud often requires analysing transaction flows across multiple institutions. Similarly, in healthcare, patient interactions are fragmented across hospitals, clinics, and insurance companies, forming an interconnected yet distributed network. In supply chain management, companies typically have visibility into their direct suppliers and customers, but the complete supply chain network spans many interdependent organisations. In all of these cases, the full value of the network data cannot be realised without aggregating subgraphs from multiple sources.  
For data consumers, aggregating subgraphs from multiple sources is essential for training reliable GNN models, enabling applications like fraud detection, personalized healthcare, and supply chain risk analysis.

In this paper, we advance existing research by exploring the pricing of individual data points within subgraphs for GNN training. This introduces a distinct challenge, as the value of any given node to the model's performance is highly dependent on its structural role and connectivity within the broader network. To address this, we aim to develop pricing mechanisms that capture the marginal contribution of each data point, taking into account both its local features and its position within the global network structure. Notably, this is the first work to tackle the problem of pricing graph-structured data in a model marketplace for GNNs.

In the following, we list our main contributions:
\begin{itemize}[leftmargin=*]
\item We propose {\em Structural Importance based Model Trading} (SIMT), a novel model marketplace framework for GNNs that integrates two phases: {\em data procurement} and {\em model training}. Figure~\ref{fig:framework} shows the conceptual framework of SIMT. 

\item For data procurement, we put forward a new method for assessing the importance of graph-structured data. For this we present a novel {\em marginal structural entropy} to quantify node informativeness. This method of importance assessment is integrated with an auction mechanism to select data owners and fairly compensate them based on their contributions. We prove that this mechanism is incentive compatible, individual rational, and budget feasible.

\item For model training, we introduce the method of {\em feature propagation} to impute missing feature data for unselected nodes, enabling effective learning with partial data.
We also design an {\em edge augmentation} method to enhance graph structure by adding connections involving unselected nodes, improving the GNN's ability to generalize.

\item The proposed SIMT method was evaluated on five well-established benchmark datasets, and consistently outperformed four baseline mechanisms in node classification tasks. SIMT achieved up to a 40\% improvement in MacroF1 and MicroF1 scores compared to the Greedy and ASCV methods, demonstrating its superior performance under various budget constraints. 
\end{itemize}

\begin{figure*}
    \centering
    \includegraphics[width=0.8\linewidth]{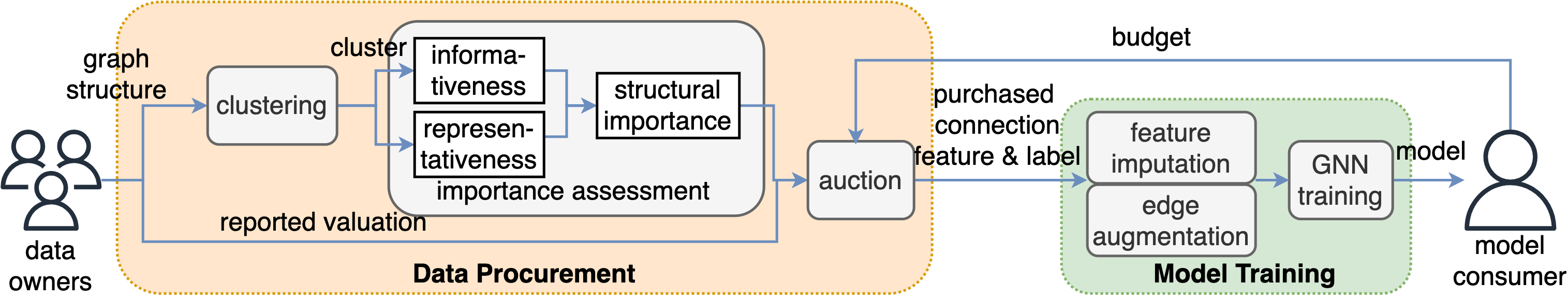}
    \caption{The framework of structural importance-based model trading (SIMT) mechanism.}
    \label{fig:framework}
\end{figure*}

\section{Related work}

Data pricing has been extensively studied in two main contexts: {\em data marketplaces} and {\em model marketplaces}. 


\noindent {\bf Data Pricing in Data Marketplaces.}
In data marketplaces, pricing mechanisms revolve around trading raw datasets or simple queries. Previous work has focused on pre-purchase decisions, where data is evaluated before it is accessed, which aligns with our setting. For instance, the importance of datasets is often quantified by metrics such as size, as explored by \cite{kushal2012pricing}, or privacy levels, as in \cite{parra2018optimized}. Other studies, such as \cite{xu2015privacy} and \cite{jaisingh2008privacy}, assess data importance based on its utility to consumers, proposing auction mechanisms and contracts to compensate data owners accordingly.

When it comes to \emph{query-based data pricing}, metrics like privacy levels directly impact the accuracy of responses, thereby influencing data value.
For instance, \cite{ghosh2011selling,roth2012conducting} propose auction mechanisms that incorporate privacy in queries, while \cite{ligett2012take} introduces a take-it-or-leave-it contract for count queries. Further work by \cite{dandekar2012privacy,zhang2020selling} expands these ideas to linear predictor queries and broader query settings.

In data marketplaces, data importance is often easily quantifiable using metrics like size or privacy levels. However, in the context of \emph{model marketplaces}, the contribution of individual data points to machine learning model performance is more complex and requires novel pricing methods.


\noindent {\bf Data Pricing in Model Marketplaces.} 
In model marketplaces, data pricing is typically based on how much a dataset improves a machine learning model's performance. \cite{abernethy2015low} introduced a theoretical framework for data pricing that balances budget constraints with model performance. Subsequent works, such as those by \cite{agarwal2019marketplace} and \cite{ghorbani2019data}, assume the model's benefit is known and focus on fairly distributing rewards among data owners. A common method for this is the \emph{Shapley value} \cite{shapley1951notes}, which compensates each data owner based on their contribution to the model.

Various studies have refined the utility function used in Shapley value calculations by incorporating additional factors. \cite{ghorbani2019data} and \cite{jia2019efficient}, for example, include $K$-nearest neighbors and privacy considerations in their utility designs. \cite{liu2020dealer} builds on this by extending the Shapley value framework to model marketplaces. Other research, such as \cite{sim2020collaborative} and \cite{ohrimenko2019collaborative}, explores utility design in \emph{collaborative machine learning} scenarios, where data owners also serve as model consumers. In these cases, utility is defined either as the sum of the model's value to the owner and its marginal value to others \cite{ohrimenko2019collaborative}, or through metrics like information gain \cite{sim2020collaborative}. \cite{xu2021gradient} and \cite{hu2020trading} further define utility based on the cosine similarity of parameters or the privacy leakage of shared model parameters.

A common limitation of these works is that they often require training models on the entire dataset before compensating data owners. In practice, this assumption is often unrealistic, as data owners are usually hesitant to contribute their data upfront without proper guarantees or compensation.


A more realistic setting, which is closer to our approach, has been explored by studies \cite{cong2020vcg,zhang2021incentive,sun2022profit}. \cite{cong2020vcg} assume that data importance is known and apply a \emph{VCG auction mechanism} to select and compensate data owners. \cite{zhang2021incentive} propose an auction mechanism that incorporates the reputation of data owners as a reflection of their contribution, while \cite{sun2022profit} design an auction that selects data owners based on their privacy requirements. Although these approaches offer valuable insights, they rely on exogenous metrics, such as reputation or privacy, which are often difficult to obtain or may not accurately reflect the intrinsic value of data for model training.

In contrast, our work proposes a novel method to \emph{measure data importance} without direct data inspection. By focusing on the structural properties of graph data and using techniques like structural entropy, we aim to create a fair and effective data pricing mechanism that overcomes the limitations of previous methods.


\noindent {\bf Comparison with FL and AL.} While {\em Federated Learning} (FL) and {\em Active Learning} (AL) are well-known paradigms for training models with distributed data, our approach differs in key ways.
{\bf (1)} In FL, each data owner trains a local model on private data, which is then aggregated into a global model while preserving privacy \cite{zhang2021survey}. SIMT, by contrast, does not require data owners to train models. Instead, data is directly provided to a central model, allowing for optimizations like data augmentation that are not possible in FL's gradient-based aggregation. This eliminates the computational burden on data owners and allows for more flexible model improvements. {\bf (2)} In AL, the model iteratively queries data points to refine learning, typically in multiple rounds \cite{ren2021survey}. SIMT, however, collects data in a single round, reducing overhead and cost. Furthermore, while AL assumes access to unlabeled data with labels provided iteratively \cite{cai2017active,zhang2021alg}, SIMT addresses the real-world challenge of compensating data owners, ensuring they are fairly rewarded for their contributions upfront.

\section{Problem formulation}
\subsection{Model marketplace for graph data}
We consider a model marketplace where a \emph{model consumer} interacts with multiple \emph{data owners} to trade graph-structured data. This data is distributed among the various data owners. Let the overall graph be represented as an attributed graph $\mathcal{G} \coloneqq (V, E, \mat{X}, \vect{y})$, where $V$ is the set of nodes, representing individual data subjects, $E \subseteq V \times V$ is the set of edges. $\mat{X} \in \mathbb{R}^{n \times m}$ is the feature matrix, where  $n$ is the number of nodes and $m$ is the dimensionality of the feature vector, and $\vect{y} \in \mathbb{R}^n$ is the label vector, where each entry corresponds to a label for each node. 
The adjacency matrix and normalised adjacency matrix of $\mathcal{G}$ are denoted as $\mat{A}$ and $\tilde{\mat{A}}$, resp.

For each node $v \in V$, let $\vect{x}_v \in \mathbb{R}^m$ and $y_v$ represent the feature vector and label of node $v$, resp. $N_v \subseteq V$ represents the set of neighbours of node $v$, and $d_v \coloneqq |N_v|$ denote the degree of node $v$.  

The graph $\mathcal{G}$ is distributed among multiple data owners, each controlling a subgraph. Let $O$ denote the set of data owners, where $o \coloneqq |O|$ represents the total number of data owners. For each data owner $i \in O$, the \emph{subgraph held by owner $i$} is represented as $\mathcal{G}_i \coloneqq (V_i, E_i, \mat{X}[V_i], \vect{y}[V_i])$, where $V_i \subseteq V$ is the set of nodes controlled by data owner $i$, $E_i = E \cap (V_i \times V_i)$ is the set of edges between nodes in $V_i$, $\mat{X}[V_i]$, and $\vect{y}[V_i]$ are the feature matrix, and label vector induced by the nodes in $V_i$, resp. Let $n_i \coloneqq |V_i|$ be the number of nodes in subgraph $\mathcal{G}_i$. 

Denote the edges within subgraphs as $\Dot{E}$ and the edges between the subgraphs as $\Ddot{E}$. We assume that $\Dot{E}\cap \Ddot{E}=\varnothing$ and then $E=\Dot{E}\cup \Ddot{E}$.
We also assume that the internal structure of each subgraph, including the features and labels of nodes, is private to the corresponding data owner. However, the connections between subgraphs (i.e., $\Ddot{E}$ the edges connecting nodes from different subgraphs) are known by the model consumer. 
Data owners are willing to sell the feature, label, and connection data for the nodes they control. 

Each data owner $i\in O$ attaches a valuation to her attribute and label data of a single node, denoted by $\theta_{i} \in \Theta$, where $\Theta$ is the set of all possible valuations. 
The valuation $\theta_{i}$ indicates the minimum payment required by the data owner to allow the use of the attribute and connection data of a single node for model training. The valuation $\theta_i$ is privately known only to the data owner, but they may report a different valuation $\theta_i'\neq \theta_i$ if it serves their interests. 
We assume that each data owner values all their data subjects {\em equally}, implying that the total valuation is linearly dependent on the number of data records. Let $\vect{\theta}_i$ be $i$'s valuation vector for all nodes, i.e., $\vect{\theta}_i \coloneqq (\theta_{i,1},\ldots,\theta_{i,n_i}) =  \theta_i \cdot \vect{1}$, where $\theta_{i,v}$ is the valuation of $i$ for node $v$. 
The valuation of all data owners form a valuation matrix, denoted by $\mat{\theta}$, which is the concatenation $\vect{\theta}_1 \mid  \cdots \mid \vect{\theta}_o \in \Theta^n$.  
The model consumer has a {\em budget}, denoted by $\beta \in \R^+$, for buying the prediction model trained on structure and attribute data. 

The model marketplace involves designing a mechanism that procures the attribute/structure data from data owners, and train a GNN model for the model consumer.

\subsection{Incentive mechanism} 
\begin{definition}
    A {\em mechanism} $M$ consists of two functions, $(\pi(\cdot),$ $p(\cdot))$, where $\pi \colon \Theta^n \to \{0,1\}^n$ is an {\em allocation function} and  $p \colon \Theta^n \to \R^n$ is a {\em payment function}. 
\end{definition}

Given a set of data owners and a model consumer, the mechanism takes the reported valuation $\mat{\theta}'\in \Theta^n$ as input, and outputs allocation result and payment result.
The allocation function and the payment function determine which nodes are selected for model training and how much to pay for the data owners, resp. 
We write the allocation result $\mat{\pi}(\mat{\theta}')$ as $(\vect{\pi}_1(\mat{\theta}'), \ldots,\vect{\pi}_o(\mat{\theta}'))$ and the payment result $\mat{p}(\mat{\theta}')$ as $(\vect{p}_1(\mat{\theta}'),\ldots,\vect{p}_o(\mat{\theta}'))$, where each $\vect{\pi}_i(\mat{\theta}'), \vect{p}_i(\mat{\theta}')$ is a $n_i$-dimensional vector with each element $\pi_{i,v},p_{i,v}$ being an allocation and payment for $i$'s node $v$. The allocation and payment of node $v$ give data owner $i$ a utility $u_{i,v}(\mat{\theta}') = (p_{i,v}(\mat{\theta}')-\theta_{i,v})\pi_{i,v}(\mat{\theta}')$. The utility of data owner $i$ is $u_i=\sum_{v\in V_i} u_{i,v}$. 
Once a node is selected, its connection, feature and label data are used for model training. 

Let $\mat{\theta}_{-i}$ denote the valuation of all data owner but $i$ and $\Theta_{-i}$ denote the set of all possible $\mat{\theta}_{-i}$. 
A mechanism $M$ should satisfy: 
\begin{itemize}[leftmargin=*]
\item {\em Incentive Compatible} (IC): Each data owner $i \in O$ gains maximum utility when truthfully reporting her valuation, i.e., 
$u_i(\theta_i,\mat{\theta}_{-i}) \geq u_i (\theta_i',\mat{\theta}_{-i}), \ \forall \theta_i,\theta_i'\in \Theta, \forall \mat{\theta}_{-i}\in \Theta_{-i}.$
\item {\em Individual Rational} (IR): Each data owner $i \in O$ gains an non-negative utility when participating in the mechanism, i.e.,  \\
$u_i(\theta_i,\mat{\theta}_{-i}) \geq 0, \ \forall \theta_i \in \Theta, \forall \mat{\theta}_{-i}\in \Theta_{-i}.$
\item {\em Budget Feasible} (BF): Total payment given to all data owners is not exceed the budget $\beta$, i.e., 
 $\sum_{i\in O} \vect{p}_i \vect{\pi}_i\leq \beta$.
\end{itemize}

\subsection{Graph neural network models} 

We use GNN as the prediction model. 
Given a graph, GNN predicts the node labels by stacking multiple layers. Let $L$ be the number of layers in a GNN model. 
The main idea is to iteratively aggregate the feature information of each node from its neighbours. Specifically, given an attributed graph $\G=(V, E, \mat{X}, \vect{y})$, and a GNN with $L$ convolution layers, at a layer $\ell \leq L$, the feature embedding $\vect{h}_v^\ell$ of node $v\in V$ is generated through aggregation and update:
\begin{itemize}[leftmargin=*]
    \item Aggregation: aggregate the feature embeddings $\vect{h}_u^\ell$ of all neighbours $u$ of $v$ by an aggregate function such as mean and sum, with trainable weights, i.e., $\vect{n}_v^\ell \coloneqq \text{Aggregator}^\ell \left(\{\vect{h}_u^\ell, \forall u\in N_v\} \right)$. 
    \item Update: update the feature embedding $\vect{h}_v^{\ell+1}$ at the next layer by an update function of the embedding $\vect{h}_v^{\ell}$ and the aggregated embeddings $\vect{n}_v^\ell$, i.e., $\vect{h}_v^{\ell+1}\coloneqq\text{Updater}^\ell\left(\vect{h}_v^\ell,\vect{n}_v^\ell\right)$. Initially, the feature embedding of node $v$ is its feature vector, i.e., $\vect{h}_v^0 \coloneqq \vect{x}_v$. 
\end{itemize}

\subsection{Optimisation problem}
As discussed in the Introduction, a key issue in determining compensation for data owners in a model marketplace is assessing the importance of their data to model training without direct inspection. 
To summarise, the problem in this paper is:

Given the model marketplace with a model consumer and several data owners, we, as a {\em data broker}, aim to design a mechanism that procures the attribute/structural data from data owners, and train a GNN model for the consumer with the following subgoals:
\begin{itemize}[leftmargin=*]
    \item assessing the importance of data to model training without disclosing the feature and label data;
    \item optimising GNN performance within the budget; and
    \item ensuring the mechanism is IC, IR, and BF.
\end{itemize}

More formally, 
let $\circ$ denote the Hadamard product operator, which selectively includes elements from $\Dot{E}, \mat{X}, \vect{y}$ according to the indicator vector $\vect{\pi} = (\vect{\pi}_1, \dots, \vect{\pi}_o)$. 
We define $f_{GNN}(\cdot)$ as the output of a GNN model trained on a selected subset of the data with the known $\Ddot{E}$.
The problem in the paper can be formulated as:
\begin{align*}
\min \quad & ||\vect{y}- f_{GNN}(\vect{\pi} \circ \G)|| \\
\text{s.t. \quad}  & u_i(\theta_i,\mat{\theta}_{-i}) \geq u_i (\theta_i',\mat{\theta}_{-i}), \ \forall \theta_i,\theta_i'\in \Theta, \forall \mat{\theta}_{-i}\in \Theta_{-i}. && \text{(IC)}  \\
 &  u_i(\theta_i,\mat{\theta}_{-i}) \geq 0, \ \forall \theta_i \in \Theta, \forall \mat{\theta}_{-i}\in \Theta_{-i} && \text{(IR)} \\
 & \sum_{i\in O} \vect{p}_i \vect{\pi}_i\leq \beta  && \text{(BF)}  \\
& \vect{p}_i\geq 0, \pi_{i,v}\in \{0,1\} \quad  \forall i\in O, \forall v\in V_i
\end{align*}

\section{Proposed method}

In this section, we propose a mechanism that procures the most contributing data and trains a GNN model using the procured data. 
By considering the correlation between graph structure, features, and labels,
we leverage the graph structure to offer insights into the contribution of the associated data. Then, we combine this contribution assessment with the data  owners' valuation in an auction mechanism to select the most cost-effective data. Subsequently, we augment the procured data using feature imputation and edge augmentation and use the augmented data to train a two-layer GNN model, which is returned to the model consumer. The overall framework is shown in Figure~\ref{fig:framework}.

\subsection{Structural importance} \label{sec:structural-importance}
We begin by evaluating the importance of data to model training without inspecting the feature and label data. 
Our solution is motivated by the observation that the structure of a graph often encodes valuable information about its features and labels. According to the well-known homophily assumption, nodes with similar features and labels are more likely to be closely connected \cite{tang2023generalized,mcpherson2001birds}. This is further validated by our case studies on five real-world graphs. We analyse the connections both within and between classes, and the results show that the number of edges within the same class is substantially higher than between different classes, as illustrated in Figure~\ref{fig:label_structure_correlation}. The strong correlation between graph structure and the associated features and labels motivates our approach to leverage the graph structure as auxiliary information in the data selection process.
Thus we propose to use the {\em structural importance} of data owner to represent her data importance.   
\begin{figure}[h]
    \centering
    \includegraphics[width=0.85\linewidth]{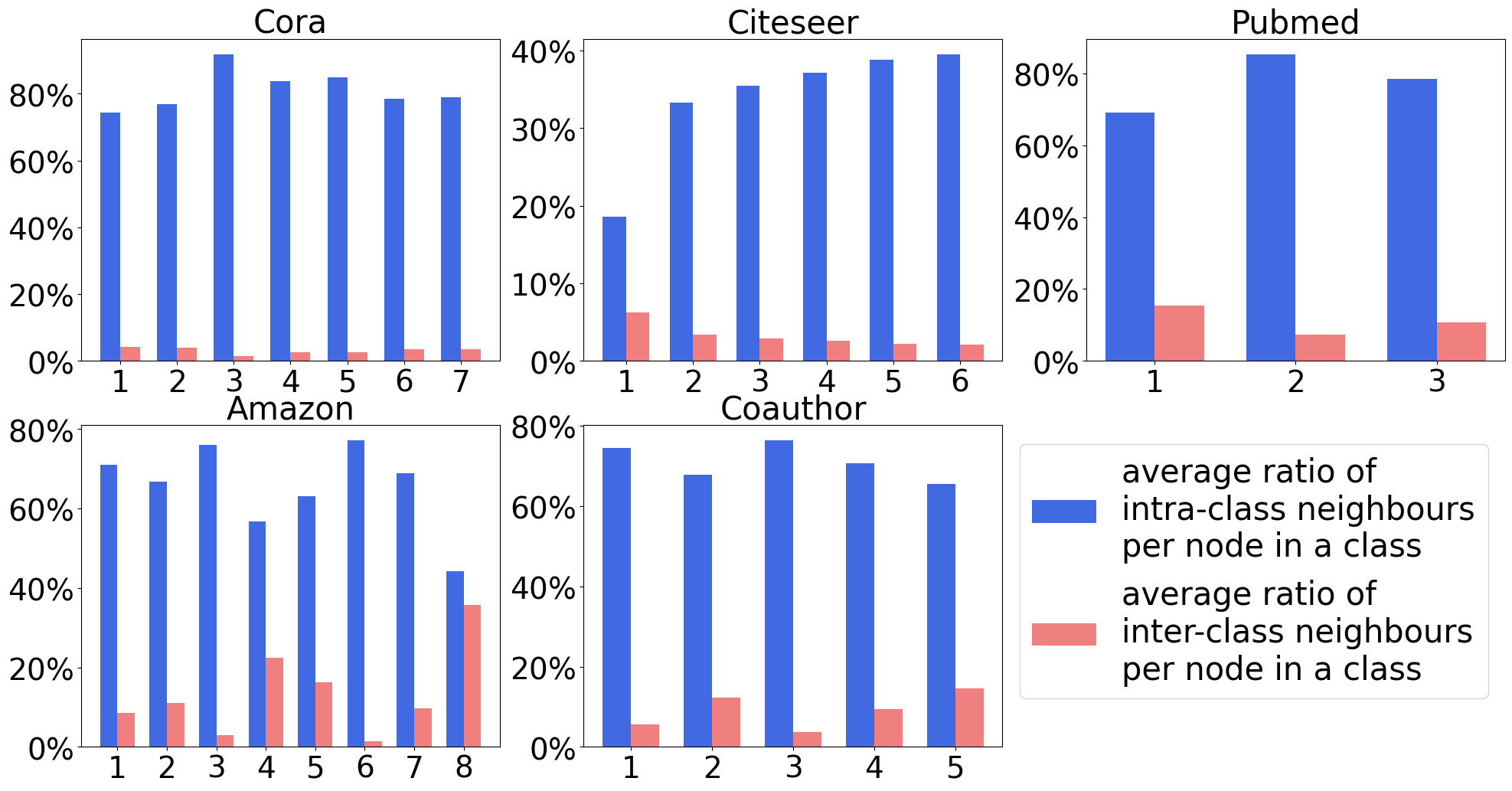}
    \caption{Proportion of intra-class and inter-class edges} 
    \label{fig:label_structure_correlation}
\end{figure} 

The sample dataset for model training should be both informative, reducing the uncertainty of the model, and representative, capturing the overall pattern of the entire dataset \cite{zhang2021alg}.  Therefore, we measure the structural importance of data owners in terms of informativeness and representativeness. 
Nevertheless, determining the informativeness and representativeness of nodes in a graph often relies on the true classes and node features, which are not available due to the absence of true feature and label data. To address this, we first use {\em structural clusters} to approximate the true classes. With this clustering, we propose the notion of {\em marginal structural entropy} to quantifies informativeness, and deploy PageRank centrality to quantify representativeness. 

\smallskip
{\em Structuring clustering. } 
A crucial tool for the structural clustering is {\em structural entropy}. 
Let $G=(V,E)$ represent a graph without attributes. 
Suppose $P \coloneqq \{C_1,C_2,\ldots,C_T\}$ is a partition of $V$, where each $C_t$ is called a {\em cluster} and $T$ is the number of clusters. {\em Structural entropy} of $G$ relative to $P$ captures the information gain due to the partition. 
For each cluster $C_t\in P$, write $d_t$ as the sum of degrees $d_v$ of all nodes $v\in C_t$.
Write $g_t$ as the number of the edges between the nodes in $C_t$ and those in the other clusters. 
The {\em structural entropy} \cite{liu2019rem} of $G$ relative to $P$ is
\begin{equation*}
\mathcal{H}_P(G) =-\sum_{t=1}^T \frac{d_t-g_t}{2|E|}\log\frac{d_t}{2|E|}.
\end{equation*}        
A greater value of $\mathcal{H}_P(G)$ means that the corresponding partition $P$ gains more information about the graph $G$ and thus $P$ is preferable.  

Given that, we would like to obtain a good partition by maximising the structural entropy $\mathcal{H}_P(G)$.
Unfortunately, maximising the structural entropy $\mathcal{H}_P(G)$ is NP-hard \cite{li2016structural,wang2023user}. As an alternative, we propose an algorithm $\mathsf{Clustering}(G)$ that harnesses the power of unsupervised GCN models \cite{wang2023user} to obtain a partition $P$.  Specifically, $\mathsf{Clustering}(G)$ first employs the Singular Value Decomposition (SVD) \cite{brunton2022data} to generate spectral features, and a classical Variational Graph Auto-Encoder (VGAE) model \cite{kipf2016variational} with reconstruction loss taking the generated spectral features as input to learn node embeddings. 
Using the obtained node embedding, $\mathsf{Clustering}(G)$ then trains a linear classifier to get a partition by maximising structural entropy. 
%

\paragraph{Structural informativeness.}
Given the learned clustering $P$, we propose the notion of {\em cluster-based marginal structural entropy} to measure the structural informativeness of data owners. Basically, the marginal structural entropy captures the information gain of a node to the clustering. The lower the marginal structural entropy of a node has, the more uncertainty this node has, and thus more information the node's data will capture.  
More formally, we define the marginal structural entropy of node $v$ as the information gain due to existence of $v$, i.e., the difference between the structural entropy of graph $G$ relative to $P$ and that of $G$ without node $v$ relative to $P'$.
Then the {\em normalised marginal structural entropy} $\epsilon_v$ of $v$ is the normalised difference in structural entropy with partition $P$ and another partition $P'$ that moves $v$ out of its cluster $C_t$, i.e., $\epsilon_v = (\mathcal{H}_P(G) - \mathcal{H}_{P'}(G) )/ \mathcal{H}_P(G)$. 
Let $n_{v,t}$ be the number of nodes that are incident to $v$ and belong to $C_t$. After calculation, we have the following (see App.~\ref{app:calculation} for the detailed calculation):  

\begin{definition}
The {\em normalised marginal structural entropy} of node $v \in C_t$ to structural entropy $\mathcal{H}_P(G)$ is 
\begin{equation*}
    \epsilon_{v} = \frac{(d_t-g_t)\log\frac{d_t}{d_t-d_v}+2n_{v,t}\log\frac{d_t-d_v}{2|E|}}{(d_t-g_t)\log\frac{d_t}{2|E|}}.
\end{equation*}
\end{definition}


\noindent A lower normalised marginal structural entropy means more structural uncertainty of $v$, making $v$ more informative. 

\paragraph{Structural representativeness.} 
We use a classical structural centrality measure, PageRank \cite{ma2008bringing}, to quantify the structural representativeness of a node. We opt for PageRank centrality due to its superior performance compared to other centrality measures, as validated in App.~\ref{app: centrality}.  The higher the PageRank centrality of a node has, the more representative the node is. 
Let $\gamma \in (0,1)$ denote the damping factor, which controls the probability of following links.
The PageRank centrality $\rho_v$ of node $v$ is:
$$
    \rho_v=\gamma \sum_{u\in N_v}\frac{\rho_u}{|N_u|}+\frac{1-\gamma}{|V|}.
$$

\paragraph{Structural importance score. } Given the clustering $P$,  the entropy $\epsilon_v$ and the PageRank centrality $\rho_v$ of each node $v$, we define the structural importance score. Specifically, we first sort all nodes in their own cluster by their entropy and PageRank values, resp. Nodes are sorted in ascending order by entropy (as lower entropy indicates higher informativeness) and in descending order by PageRank (as higher PageRank indicates greater representativeness). This ensures that more informative and representative nodes are prioritised.
For a node $v$ in cluster $C_t$, let $\text{rank}^{\text{entr}}_v$ denote its rank by entropy and $\text{rank}^{\text{pr}}_v$ denote its rank by PageRank.  We then define node $v$'s informativeness and representativeness based on these rankings as
$$\phi^{\text{info}}_v \coloneqq \frac{\text{rank}^{\text{entr}}_v}{|C_t|} \text{, and } \phi_v^{\text{rep}}\coloneqq \frac{\text{rank}^{\text{pr}}_v}{|C_t|}, \text{ resp.}$$

Finally, we define the structural importance score of a node. 
Following the approach in \cite{cai2017active,zhang2021alg}, we introduce a parameter $\alpha$ to balance representativeness and informativeness.
Intuitively, representative data helps to learn general classification patterns, while informative data is used to refine the classification boundaries.
Therefore, when the budget $\beta$ is relatively small compared to the overall valuations of data owners and the partition $P$ is complex (i.e., when $T$ is large), prioritising representative data is crucial to learning the general classification patterns. On the other hand, when the budget is relatively large and the partition is simpler, a small amount of representative data is sufficient to capture the overall pattern, allowing us to focus on acquiring more informative data to further refine the classification. 
More formally, given the average valuation $\overline{\theta}$, defined as the average of the upper and lower bounds of data valuations, we set $\alpha = \frac{1}{2}(1+\frac{\beta}{n\overline{\theta}})^{-T}$. 
The {\em structural importance score} of node $v$ is then defined as 
\begin{equation}
\label{eqn:score}
  \phi_v \coloneqq (1-\alpha)\phi^{\text{rep}}_v+\alpha \phi^{\text{info}}_v.  
\end{equation}
    

\subsection{Model trading mechanism} \label{sec:model-trading}

We propose a model trading mechanism, named {\em Structural importance based model trading (SIMT) mechanism}, which consists of two phases: (1) data procurement phase selects the most cost effective data owners, and (2) model training phase trains a GNN model on the procured data; See the workflow of SIMT in Figure~\ref{fig:framework} and the algorithm in Alg.~\ref{alg:trading}.

\paragraph{Phase 1. Data procurement.} In Phase 1, SIMT takes the attributed graph $\G$ and the valuation vector $\mat{\theta}$, and the budget $\beta$ as inputs and returns an allocation result $\vect{\pi}$ and a payment result $\vect{p}$. Firstly,  $\mathsf{Clustering}(G)$ returns a clustering $P$ of the nodes $V$ after training. Given the clustering $P$ with $T$ clusters,  the mechanism computes the structural importance score $\phi_v$ of each node $v \in V$ using Equation~\eqref{eqn:score}, which represents the importance of the node. Then an auction is conducted in each cluster $C_t\in P$ with budget $\beta/T$. In the auction for each $C_t$, nodes in $C_t$ are sorted in descending order based on the ratio $\frac{\phi_v}{\theta_{i,v}}$, where data owner $i$ owns node $v$.  
The mechanism selects the most cost-effective $k$ data until the total payment exceeds the allocated budget. The payment to data owner $i$ for node $v$ is $p_{i,v} = \min\{ \frac{\beta}{T} \frac{\phi_v}{\sum_{u=1}^k \phi_u}, \frac{\theta_{j,w}}{\phi_{w}} \phi_v\}$, where $j$ is the first data owner who has not had any data selected and $w$ is her first data in the order (if such data owner $j$ does not exist, we set $p_{i,v}$ as $\min\{\beta \phi_v/ (T\sum_{u=1}^k \phi_u), \tilde{\theta} \phi_v/\phi_{k+1}\}$, where $\tilde{\theta}$ is the upper bound of $\Theta$).
The total payment to data owner $i$ is $\sum_{v\in V_i, v\leq k}p_{i,v}$.


\begin{algorithm}[t]
\caption{The SIMT mechanism}
\label{alg:trading}
\footnotesize
\begin{algorithmic}[1]
\Require Attributed graph $\G$, data owners $O$, valuation vector $\vect{\theta}$, and budget $\beta$
\Ensure Allocation $\vect{\pi}$, payment $\vect{p}$, and trained model $f_{\text{GNN}}$ 
\State let $G=(V,\Ddot{E})$ represent the known subgraph of $\G$ without attributes.
\\
{\bf Phase 1: Data procurement}
\State get a partition $P \gets \mathsf{Clustering}(G)$   \label{ln:Clustering}  
\State compute the structural importance score $\phi_v$ for each $v\in V$ according to $P$
\State initialise $\vect{\pi}=\boldsymbol{0},\vect{p}=\boldsymbol{0}$
\For{each cluster $C_t \in P$}
\State sort the nodes $v \in C_t$ by $\frac{\phi_v}{\theta_{i,v}}$ in a descending order 
\State find the largest $k$ such that $\theta_k \leq \frac{\phi_k}{\sum_{u\leq k} \phi_u} \frac{\beta}{T}$
\For{$v\leq k$}
\State Let $i\in O$ be the data owner of node $v$
\State set $\pi_{i,v}=1$, $p_{i,v} = \min\{ \frac{\beta}{T} \frac{\phi_v}{\sum_{u=1}^k \phi_u}, \frac{\theta_{j,w}}{\phi_{w}} \phi_v\}$
\EndFor 
\EndFor 
\State procure data to get $\mat{X}_s,\vect{y}_s$, update $G$ and normalized adjacency matrix $\tilde{\mat{A}}$
\\
{\bf Phase 2: Model training}
\State initialise $\mat{X}'\gets [\mat{X}_s,\mat{0}_u]^ \intercal$ 
\While{ $\mat{X}'$ has not converged} 
\State $\mat{X}'\gets \tilde{\mat{A}} \mat{X}'$
\State $\mat{X}'\gets [\mat{X}_s,\mat{X}'_u]^ \intercal$
\EndWhile
\State do edge augmentation on $G$ and get $\overline{G}$
\State $f_{\text{GNN}}\gets \mathsf{Train} (\overline{G
},\mat{X}',\vect{y}_s) $
\end{algorithmic}
\end{algorithm}



\paragraph{Phase 2. Model training.} Given the allocation result $\vect{\pi}$, the model training phase uses the connections, features and labels of the selected data owners to train a GNN model. 
However, due to budget constraints, only a subset of features and connections can be purchased, which may not be sufficient for training a robust GNN model. 
To address this, we first impute the missing node features and augment the missing edges before training.
After acquiring the features from the selected nodes, we apply the feature propagation algorithm \cite{rossi2022unreasonable} to infer the features of the unselected nodes, producing a new feature matrix $\mat{X}'$.
Additionally, we use the $G(|V|,|E|)$ Erdos-Renyi (ER) model \cite{bollobas1998random} to generate missing edges, where $|V|$ and $|E|$ are determined by the edge density of the known graph. Then we incorporate contrastive learning \cite{oord2018representation} to mitigate the randomness introduced by the ER model, resulting in an augmented graph $\overline{G}$. 
The GNN training algorithm then takes the augmented graph $\overline{G}$, the new feature matrix $\mat{X}'$ and labels $\vect{y}_s$ as input and returns a GNN model to the consumer.  


{\em Node feature imputation.} 
Features are crucial when training GNN \cite{taguchi2021graph,spinelli2020missing}, and we apply a feature imputation method to the procured data to address missing values. Among various feature imputation methods, we choose the Feature Propagation algorithm due to its strong convergence guarantees, simplicity, speed, and scalability \cite{rossi2022unreasonable}.
We use subscripts $s$ and $u$ to denote the selected and unselected nodes, resp. Write $\mat{X}=[\mat{X}_s,\mat{X}_u]^ \intercal$ and $\vect{y}=[\vect{y}_s,\vect{y}_u]^ \intercal$. 
Also, we write the normalised adjacency matrix $\tilde{\mat{A}}$ and the graph Laplacian $\Delta$ of $\overline{\G}$ as
$\tilde{\mat{A}}=\begin{bmatrix}
\tilde{\mat{A}}_{ss} & \tilde{\mat{A}}_{su} \\
\tilde{\mat{A}}_{us} & \tilde{\mat{A}}_{uu} 
\end{bmatrix},
\Delta=\begin{bmatrix}
\Delta_{ss} & \Delta_{su} \\
\Delta_{us} & \Delta_{uu} 
\end{bmatrix}$, resp.
In the feature imputation process, the feature matrix $\mat{X}'$ is initialised with the known feature $\mat{X}_s$ and a zero matrix $\mat{0}_u$ for the unselected nodes.
The feature matrix $\mat{X}'$ is then iteratively updated as follows:
$\mat{X}^{(t)}=\begin{bmatrix}
 \mat{I} & \mat{0} \\
 \tilde{\mat{A}}_{us} & \tilde{\mat{A}}_{uu} 
 \end{bmatrix} \mat{X}^{(t-1)},$
where this process continues until the feature matrix converges.
The steady status of the feature matrix is \cite{rossi2022unreasonable}: $\lim_{t\to \infty} \mat{X}^{(t)} = \begin{bmatrix}
    \mat{X}_s\\
    -\Delta_{ss}^{-1} \tilde{\mat{A}}_{us} \mat{X}_s
\end{bmatrix}$.


\paragraph{Edge Augmentation.}
Given the critical role of message passing in GNNs, the absence of certain edges may impede this process, leading to sub-optimal model performance.  To alleviate this issue, we introduce augmented edges to  enhance message passing. 
Specifically, we employ the ER model \cite{bollobas1998random} to generate edges for data owners with multiple unselected nodes.
However, the introduction of augmented edges may inadvertently introduce noise, which could mislead the model by learning from incorrect connections. To counteract this, we integrate contrastive loss \cite{oord2018representation}, denoted by $L_{\text{ctr}}$, into the GNN training process. This loss function encourages the model to maximise the similarity between the augmented graph and the original (non-augmented) graph views.
Given a graph $\G$ and an augmented graph $\overline{\G}$, let $\vect{h}_v$ and  $\vect{h}_v'$ be the feature embeddings in $\G$ and  $\overline{\G}$, resp. 
The contrastive loss of a node $v$ is:
\begin{equation*}
    L_{\text{ctr}}(v)=-\log \frac{\exp(\vect{h}_v\cdot \vect{h}_v'/\tau)}{\sum_{u\in V}\exp(\vect{h}_v\cdot \vect{h}_u'/\tau)}, 
\end{equation*}
where $\tau$ represents the temperature parameter, which scales the similarities between the embeddings $\vect{h}_v$ and $\vect{h}_v'$. 



\subsection{Analysis}

Now we show that SIMT satisfies IC, IR and BF properties. Additionally, we analyse its time complexity.

\begin{theorem}
The SIMT mechanism is incentive compatible, individual rational and budget feasible. 
\end{theorem}


\begin{proof}
We first show that the mechanism is IR. In each auction of cluster $C_t$, the utility that a data owner $i$ obtains from a node $v>k$ when she truthfully reports is $u_{i,v}(\vect{\theta})=\theta_{i,v} \pi_{i,v}(\theta)-p_{i,v}(\theta)=0-0\geq 0$; the utility that a data owner $i$ obtains from $v\leq k$ is $u_{i,v}(\theta)=p_{i,v}(\theta)-\theta_{i,v} \pi_{i,v}(\theta)=\min\{ \nicefrac{(\beta\phi_v)}{(T\sum_{u=1}^k \phi_u)}, \nicefrac{(\phi_v\theta_{j,w})}{\phi_{w}} \}-\theta_{i,v}\times 1 \} \geq 0$. 
Therefore, the utility of data owner $i$ is $
\sum_{v\in V_i} u_{i,v} \geq 0$. 
Also, SIMT satisfies BF. The total payment in cluster $C_t$ is \\ $\sum_{v=1}^k \min\{ \nicefrac{(\beta \phi_v)}{(T\sum_{u=1}^k \phi_u)}, \nicefrac{(\theta_{j,w}\phi_v}{\phi_{w}} \} \leq \sum_{u=1}^k \phi_u \times \nicefrac{\beta}{(T \sum_{u=1}^k \phi_u)} = \nicefrac{\beta}{T}$. Then the total payment in $T$ clusters is $\leq \beta$, which shows BF. 

Next we show IC. Each of data subject is assigned to an auction associated with a cluster $C_t$. As the assignment is independent of the reported valuation of the data owner, we just need to show that in the auction for each cluster is IC. 
In one auction, we consider an arbitrary data owner. When she report truthfully, there are two cases regarding each of her node, either being selected or not. We discuss the two cases separately. 

\noindent{\bf (1)} Consider an arbitrary node $v$ that is selected. Assume that in the ranking, the $(k+1)$-th node is possessed by data owner $j$, i.e., the $(k+1)$-th ratio is $\nicefrac{\phi_{k+1}}{\theta_{j,k+1}}$. Note that the data owner $j$ could be $i$. 
If $i$ reports a lower valuation $\theta_{i,v}'< \theta_{i,v}$ or a higher valuation $ \theta_{i,v} < \theta_{i,v}' <  \nicefrac{\theta_{j,k+1}\phi_v}{\phi_{k+1}}$, as the marginal contribution $\phi_v$ is independent from the reported valuation, the ratio $\nicefrac{\phi_v}{\theta_{i,v}'} \geq \nicefrac{\phi_{k+1}}{\theta_{j,k+1}}$ and her ranking is still in the top $k$. Further,  the payment of $i$ for $v$ is independent from $i$'s report. As a consequence, her utility of $v$ is $u_{i,v}(\vect{\theta}_{i}',\vect{\theta}_{-i}) = u_{i,v}(\vect{\theta}_i,\vect{\theta}_{-i})$. 
If she reports a even higher valuation $\theta_{i,v}' \geq  \nicefrac{\theta_{j,k+1}\phi_v}{\phi_{k+1}}$, the ratio $\nicefrac{\phi_{k+1}}{\theta_{j,k+1}} > \nicefrac{\phi_v}{\theta_{i,v}'}$. Then her allocation becomes $0$ and her utility of $v$ is $u_{i,v}(\vect{\theta}_i',\vect{\theta}_{-i}) = 0 \leq u_{i,v}(\vect{\theta}_i,\vect{\theta}_{-i})$.
    
\noindent{\bf (2)} Consider a node $v$ that is not selected. Assume that in the ranking, the $k$-th node is possessed by data owner $j$, i.e, the $k$-th ratio is $\nicefrac{\phi_{k}}{\theta_{j,k}}$. Here, $j$ could also be $i$. 
If $i$ reports a higher valuation $\theta_{i,v}' > \theta_{i,v}$ or a lower valuation $\nicefrac{\theta_{j,k}\phi_v}{\phi_{k}} \leq \theta_{i,v}' < \theta_{i,v}$, the ratio $\nicefrac{\phi_k}{\theta_{j,k}} \geq \nicefrac{\phi_v}{\theta_{i,v}'}$, and $v$'s ranking is still not among the first $k$. Then $i$'s utility of $v$ is $u_{i,v}(\vect{\theta}_i',\vect{\theta}_{-i}) = u_{i,v}(\vect{\theta}_i,\vect{\theta}_{-i})=0$. 
If $i$ reports a much lower valuation  $\theta_{i,v}'<  \nicefrac{\theta_{j,k}\phi_v}{\phi_{k}}$, and her ranking is among the first $k$. Her utility $u_i(\vect{\theta}_i',\vect{\theta}_{-i})  = \theta_i\!- \!\min\{ \frac{\beta}{T} \frac{\phi_v}{\sum_{u=1}^k \phi_u}, \frac{\theta_{j,w}}{\phi_{w}} \phi_v\} \! \leq \! 0 = u_i(\vect{\theta}_i,\vect{\theta}_{-i}).$
\end{proof}

\noindent{\bf Time complexity.}
Given a GNN, recall $m$ is the dimensions of the input and let  $m_{\text{h}}$ be the hidden layers. 
The computational complexity of a typical GNN is $O(|E| m + |V| m m_{\text{h}})$.  The computational complexity of SIMT is $O((|E| m + |V| m m_{\text{h}})+(|E| m + |V| m m_{\text{h}})+(|V|+|E|))$, where the first one terms correspond to the complexity of training the GNN,  while the last two terms account for the computation of clustering and PageRank centrality resp.

\section{Experiment}
We conduct experiments to validate the performance of proposed SIMT mechanism in terms of node classification accuracy. 
(1) To demonstrate the overall performance of SIMT, we compare it with multiple baselines under different budgets. (2) To underscore the impact of each component, we perform a detailed ablation study.

\subsection{Experiment setup} \label{sec:setup}

\paragraph{Dataset.} 
Five widely-used datasets are included in our experiments: Cora, Citeseer, Pubmed, Amazon and Coauthor \cite{kipf2016semi,cai2017active,rossi2022unreasonable,shchur2018pitfalls}. The dataset statistics are listed in Table~\ref{tab:dataset_statistics} in App.~\ref{app: dataset}.
For each dataset, we randomly sample $15\%$ of the data as test set, which remains untouched during data procurement. This set is consistent across all baselines.
Once getting the selected data from auction, we further split $80\%$ as training data and remaining $20\%$ for validation.
%
To accommodate different real-world scenarios, we follow the setup in existing studies \cite{hu2020trading,jia2019efficient,liu2020dealer,ohrimenko2019collaborative,sim2020collaborative,xu2021gradient} to 
validate SIMT on various datasets and varying hyperparameters.


\begin{table*}[]
\caption{Node classification performance under different budgets}
\label{tab:main-performance}
\resizebox{\textwidth}{!}{%
\begin{tabular}{|cc|cc|cc|cc|cc|cc|cc|c|c|}
\hline
\multicolumn{2}{|c|}{budget} & \multicolumn{2}{c|}{$50$} & \multicolumn{2}{c|}{$100$} & \multicolumn{2}{c|}{$150$} & \multicolumn{2}{c|}{$200$} & \multicolumn{2}{c|}{$250$} & \multicolumn{2}{c|}{$300$} & - & - \\ \hline
\multicolumn{2}{|c|}{metric} & MacroF1 & MicroF1 & MacroF1 & MicroF1 & MacroF1 & MicroF1 & MacroF1 & MicroF1 & MacroF1 & MicroF1 & MacroF1 & MicroF1 & \begin{tabular}[c]{@{}c@{}}average\\ accuracy\end{tabular} & $\frac{\text{ave. accuracy}}{\# \text{bought.items}}$ \\ \hline
\multicolumn{1}{|c|}{\multirow{5}{*}{\small Cora}} & Greedy & $12.3\pm 4.2$ & $21.5\pm 6.3$ & $17.6\pm 5.6$ & $29.1\pm 11.1$ & $20.6\pm 6$ & $33.9\pm 9.8$ & $23.6\pm 6$ & $36.7\pm 9.8$ & $26.5\pm 5.9$ & $40.1\pm 9.3$ & $28.0\pm 4.7$ & $41.5\pm 8.8$ & 33.8 & 0.16 \\
\multicolumn{1}{|c|}{} & ASCV & $10.9\pm 4.5$ & $21.0\pm 6.6$ & $16.2\pm 5.3$ & $29.1\pm 8.2$ & $23.8\pm 6.9$ & $36.0\pm 5.1$ & $23.7\pm 4.5$ & $36.0\pm 4.9$ & $28.1\pm 5.6$ & $38.5\pm 5.0$ & $32.4\pm 8.9$ & $43.9\pm 6.2$ & 34.1 & 0.20 \\
\multicolumn{1}{|c|}{} & Greedy(P) & $24.6\pm 10.6$ & $36.1\pm 14.4$ & $29.0\pm 9.6$ & $40.6\pm 12$ & $32.3\pm 8.5$ & $44.5\pm 10$ & $36.0\pm 9.8$ & $48.8\pm 8.4$ & $37.3\pm 9.3.5$ & $48.9\pm 8$ & $38.6\pm 10$ & $50.8\pm 8.3$ & 45.0 & 0.22 \\
\multicolumn{1}{|c|}{} & ASCV(P) & $21.9\pm 5.8$ & $34.8\pm 6.8$ & $37.2\pm 11$ & $46.9\pm 9.8$ & $47.3\pm 7$ & $53.7\pm 6.5$ & $55.3\pm 5.9$ & $61.2\pm 4.9$ & $60.3\pm 5.1$ & $65.2\pm 5.1$ & $65.0\pm 6.3$ & $68.1\pm 5.7$ & 55.0 & 0.33 \\
\multicolumn{1}{|c|}{} & \textbf{SIMT} & $\bm{36.1\pm 10.7}$ & $\bm{48.4\pm 9.5}$ & $\bm{46.7\pm 8.8}$ & $\bm{55\pm 6.7}$ & $\bm{56.8\pm 10}$ & $\bm{62.5\pm 8.6}$ & $\bm{62.2\pm 5.8}$ & $\bm{67.8\pm 3.3}$ & $\bm{63.4\pm 5.3}$ & $\bm{69.6\pm 2.8}$ & $\bm{66.7\pm 6.8}$ & $\bm{71.8\pm 4.5}$ & $\bm{62.5}$ & $\bm{0.37}$ \\ \hline
\multicolumn{1}{|c|}{\multirow{5}{*}{\small Citeseer}} & Greedy & $6.7\pm 2.7$ & $17.7\pm 7.2$ & $11.0\pm 5.3$ & $22.6\pm 8.3$ & $14.4\pm 5.1$ & $24.5\pm 7.8$ & $18.7\pm 6$ & $29.0\pm 7.8$ & $20.5\pm 7.4$ & $31.0\pm 8.1$ & $21.9\pm 6.6$ & $32.6\pm 8.3$ & 24.6 & 0.12 \\
\multicolumn{1}{|c|}{} & ASCV & $8\pm 3.4$ & $17.4\pm 7.5$ & $12.9\pm 5.7$ & $21.8\pm 7.4$ & $16.9\pm 4.8$ & $26.5\pm 4.7$ & $20.8\pm 5.2$ & $29.6\pm 6.6$ & $24.3\pm 5.3$ & $34.4\pm 5.2$ & $28.1\pm 5.2$ & $36.8\pm 6.4$ & 27.8 & 0.17 \\
\multicolumn{1}{|c|}{} & Greedy(P) & $15.3\pm 7.5$ & $24.5\pm 9.3$ & $21.5\pm 8.5$ & $30.7\pm 9.6$ & $23.9\pm 7.9$ & $33.2\pm 8$ & $27.4\pm 10$ & $36.5\pm 7.7$ & $29.8\pm 9.8$ & $39.0\pm 7.7$ & $32.9\pm 11.1$ & $42.4\pm 8.5$ & 34.4 & 0.17 \\
\multicolumn{1}{|c|}{} & ASCV(P) & $19.8\pm 7.1$ & $27.2\pm 9.5$ & $32.2\pm 5.9$ & $38.3\pm 9.2$ & $38.9\pm 4$ & $43.8\pm 6.3$ & $42.4\pm 7$ & $48.0\pm 7.6$ & $46.7\pm 2.7$ & $50.5\pm 4.5$ & $48.6\pm 4$ & $51.9\pm 6$ & 43.3 & 0.26 \\
\multicolumn{1}{|c|}{} & \textbf{SIMT} & $\bm{28.7\pm 5.8}$ & $\bm{39.4\pm 6.3}$ & $\bm{37.0\pm 5.6}$ & $\bm{47.4\pm 6.1}$ & $\bm{41.7\pm 6}$ & $\bm{50.2\pm 6.3}$ & $\bm{45.4\pm 5.2}$ & $\bm{53.4\pm 6.7}$ & $\bm{47.2\pm 5.1}$ & $\bm{55.8\pm 3.8}$ & $\bm{50.0\pm 4.9}$ & $\bm{57.9\pm 4.8}$ & $\bm{50.7}$ & $\bm{0.30}$ \\ \hline
\multicolumn{1}{|c|}{\multirow{5}{*}{\small Pubmed}} & Greedy & $15.7\pm 3.8$ & $30.8\pm 9.8$ & $16.3\pm 3.4$ & $30.8\pm 9.7$ & $19.2\pm 5.9$ & $34.0\pm 8.9$ & $20.8\pm 7$ & $35.0\pm 8.4$ & $22.6\pm 9.1$ & $37.5\pm 8.9$ & $21.2\pm 8.3$ & $35.8\pm 9.6$ & 34.0 & 0.14 \\
\multicolumn{1}{|c|}{} & ASCV & $17.2\pm 4.3$ & $33.7\pm 9$ & $17.4\pm 4.7$ & $33.9\pm 9.2$ & $17.0\pm 4$ & $33.8\pm 9.1$ & $18.1\pm 5.7$ & $34.2\pm 9.3$ & $16.8\pm 3.8$ & $33.8\pm 9$ & $18.1\pm 6.2$ & $34.5\pm 9.8$ & 34.0 & 0.16 \\
\multicolumn{1}{|c|}{} & Greedy(P) & $19.7\pm 7.7$ & $34.1\pm 9.4$ & $25.1\pm 12.1$ & $39.3\pm 10.5$ & $27.6\pm 13.3$ & $40.9\pm 11.4$ & $29.9\pm 14.9$ & $43.0\pm 13.3$ & $30.6\pm 15.4$ & $43.4\pm 13.6$ & $31.2\pm 14.4$ & $43.9\pm 12.6$ & 40.8 & 0.17 \\
\multicolumn{1}{|c|}{} & ASCV(P) & $20.0\pm 9.8$ & $35.8\pm 11.7$ & $21.4\pm 12.5$ & $36.3\pm 12.5$ & $22.1\pm 13.5$ & $36.9\pm 13.2$ & $23.3\pm 15.7$ & $37.8\pm 14.7$ & $23.4\pm 16.1$ & $37.7\pm 14.7$ & $23.9\pm 16$ & $37.9\pm 14.2$ & 37.0 & 0.17 \\
\multicolumn{1}{|c|}{} & \textbf{SIMT} & $\bm{22.7\pm 7.2}$ & $\bm{38.5\pm 7.6}$ & $\bm{29.9\pm 13.5}$ & $\bm{44.0\pm 10.6}$ & $\bm{36.7\pm 16.1}$ & $\bm{49.2\pm 11.9}$ & $\bm{42.0\pm 14.2}$ & $\bm{53.3\pm 10.5}$ & $\bm{43.1\pm 15.1}$ & $\bm{54.4\pm 11.6}$ & $\bm{50.2\pm 16.7}$ & $\bm{60.2\pm 11.1}$ & $\bm{49.9}$ & $\bm{0.25}$ \\ \hline
\multicolumn{1}{|c|}{\multirow{5}{*}{\small Amazon}} & Greedy & $6\pm 3.2$ & $15.5\pm 6$ & $8.7\pm 5.4$ & $18.2\pm 6.8$ & $9.7\pm 4.5$ & $17.8\pm 5.8$ & $10.7\pm 4.8$ & $18.5\pm 6.4$ & $11.3\pm 5.3$ & $19.6\pm 5.6$ & $13.8\pm 6.1$ & $21.3\pm 5.5$ & 18.5 & 0.08  \\   
\multicolumn{1}{|c|}{} & ASCV & $4.4\pm 1.5$ & $17.3\pm 6.6$ & $7.1\pm 3.6$ & $20.9\pm 5.3$ & $9.6\pm 5.2$ & $25.2\pm 3.4$ & $11.4\pm 7.4$ & $27.1\pm 5.6$ & $11.9\pm 7.2$ & $27.3\pm 5$ & $12.6\pm 7.2$ & $29\pm 6$ & 24.5 & 0.12  \\ 
\multicolumn{1}{|c|}{} & Greedy(P) & $19.6\pm 5.5$ & $24.9\pm 6.2$ & $23.2\pm 8$ & $29.3\pm 8.8$ & $23.4\pm 8.4$ & $29.6\pm 9.1$ & $24.8\pm 8.2$ & $31.2\pm 8.3$ & $24.9\pm 7.8$ & $31.7\pm 8.9$ & $25.4\pm 7.4$ & $32.6\pm 8.0$ & 30.0 &0.13  \\ 
\multicolumn{1}{|c|}{} & ASCV(P) & $18.2\pm 7.5$ & $31.1\pm 8.0$ & $21.7\pm 7.3$ & $34.8\pm 8.3$ & $23.9\pm 8.1$ & $36.3\pm 8.9$ & $26.7\pm 7.0$ & $41.9\pm 8.9$ & $29.7\pm 9.4$ & $43.4\pm 9.3$ & $30.2\pm 9.1$ & $44.7\pm 9.6$ & 38.7  & 0.20  \\
\multicolumn{1}{|c|}{} & \textbf{SIMT} & $\bm{30\pm 7.1}$ & $\bm{38.9\pm 7.9}$ & $\bm{38.1\pm 6.7}$ & $\bm{48.4\pm 6.6}$ & $\bm{40.7\pm 7.4}$ & $\bm{50.6\pm 6.7}$ & $\bm{44.2\pm 7.1}$ & $\bm{57.1\pm 5.9}$ & $\bm{46.2\pm 4.3}$ & $\bm{57.8\pm 3.6}$ & $\bm{51.7\pm 7.8}$ & $\bm{60.3\pm 6.5}$  & $\bm{52.2}$  & $\bm{0.28}$ \\ \hline
\multicolumn{1}{|c|}{\multirow{5}{*}{\small Coauthor}} & Greedy & $8.5\pm 5.6$ & $24.2\pm 19.2$ & $9.2\pm 5.9$ & $24.5\pm 19.4$ & $11.7\pm 6.5$ & $29.0\pm 19.8$ & $12.9\pm 7.5$ & $29.8\pm 20.2$ & $15.4\pm 9.0$ & $34.4\pm 20.2$ & $17.1\pm 9.0$ & $35.2\pm 20.1$ & 29.5  & 0.12  \\ 
\multicolumn{1}{|c|}{} & ASCV & $10.3\pm 4.3$ & $34.0\pm 17.7$ & $11.1\pm 4.8$ & $37.1\pm 17.7$ & $11.8\pm 5.1$ & $37.4\pm 17.7$ & $12.9\pm 5.3$ & $38.7\pm 16.7$ & $16.1\pm 6.4$ & $44.7\pm 14.7$ & $16.3\pm 6.9$ & $45.4\pm 15.1$ & 39.5 & 0.18 \\  
\multicolumn{1}{|c|}{} & GreedyP & $18.0\pm 8.4$ & $34.2\pm 19.3$ & $24.7\pm 11.3$ & $39.6\pm 20.0$ & $27.5\pm 12.3$ & $42.0\pm 21.0$ & $27.5\pm 12.3$ & $42.1\pm 21.4$ & $29.4\pm 13.9$ & $44.4\pm 22.3$ & $30.0\pm 14.1$ & $44.1\pm 22.2$ & 41.1  & 0.16 \\
\multicolumn{1}{|c|}{} & ASCVP & $16.8\pm 9.1$ & $43.5\pm 18.6$ & $26.2\pm 9.2$ & $54.1\pm 13.9$ & $25.6\pm 12.2$ & $51.7\pm 18.9$ & $30.7\pm 11.7$ & $57.6\pm 15.3$ & $30.8\pm 11.7$ & $58.3\pm 15.4$ & $32.1\pm 10.7$ & $59.3\pm 15.1$ & 54.1 & 0.24  \\
\multicolumn{1}{|c|}{} & \textbf{SIMT} & $\bm{24.9\pm 8.9}$ & $\bm{55.5\pm 5.6}$ & $\bm{29.8\pm 9.8}$ & $\bm{57.9\pm 8.2}$ & $\bm{33.5\pm 8.4}$ & $\bm{61.5\pm 6.8}$ & $\bm{39.1\pm 7.6}$ & $\bm{63.4\pm 6.4}$ & $\bm{42.8\pm 12.0}$ & $\bm{65.2\pm 8.0}$ & $\bm{46.7\pm 11.6}$ & $\bm{68.8\pm 6.3}$ & $\bm{62.0}$ & $\bm{0.30}$  \\ \hline
\end{tabular}%
}
\end{table*}

{\it Data valuations.}
We generate a set of random numbers to represent the data valuations. The valuations  are sampled at random i.i.d. following a series of normal distributions $\mathcal{N}(\mu,\,\sigma^{2})$. 
We get a $\mu$ drawn from $\mathcal{U}[0.8,1.2]$ for each class to capture the difference in valuations between classes. Then for each data owner, we set the valuation of each data subject as the mean of the generated valuations of her data subjects.
We set $\sigma=0.1$. The effect of different $\sigma$s on performance is investigated and the results are in App.~\ref{app:deviation}.
To ensure all valuations are non-negative, we use a resample trick \cite{burkardt2014truncated}. The generated valuations are in the range $[0,2]$. Note that when the domain is different, we could scale it into $[0,2]$. 


{\it Budget.}
We set the budget in $\{50,100,150,200, 250, $ $300\}$. Given that the data valuation range is $[0,2]$,
the number of selected data is approximately from $50$ to $300$, which is aligned with the setup of the studies on label selection e.g. \cite{zhang2021alg,cai2017active}.


{\it Structural clustering.}
Here, we give the configuration of the model used for $\mathsf{Clustering}(G)$.
We deploy SVD \cite{brunton2022data} to generate spectral node features, VGAE \cite{kipf2016variational} to learn node embeddings followed by a linear classifier to learn the partition. 
We set the hidden size as $32$, the learning rate as $0.01$, the L2 regularisation as $5\times 10^{-4}$. The total training budget is $400$ epochs.
The clustering model is initialised to solely minimise the reconstruction loss.
We repeat this process $100$ epochs to comprehensively capture the graph structure information. 
Using the obtained node embeddings, we train the linear classifier to learn a partition, maximising the structural entropy. This process is repeated $300$ epochs to obtain a robust partition.

{\it GNN model.}
We employ classical GNN models as $f_{\text{GNN}}$ to learn node classification. 
Following the configurations of \cite{kipf2016semi}, we set the hidden size as $32$, the total training budget as $200$, the learning rate as $0.01$ and the L2 regularisation as $5 \times 10^{-4}$.
The models are optimised with minimising both reconstruction loss $L_{\text{recon}}$ and classification loss $L_{\text{class}}$ on the train data.
We repeat $10$ training iterations with different random seeds and report the average performance.
To mitigate the impact of randomness in train-validation splitting,  each training iteration creates $10$ train-validation splits, trains $10$ independent models according to the split and reports the best model according to their performance on the validation set. 
Ultimately, we evaluate the model performance using the test data.
In other words, each experimental result is derived from $100$ runs.
We present the results using a GCN model and defer the exploration on the effects of different GNN architectures in App.~\ref{app:models}. 

{\it Subgraph.}
Each data owner possesses a subgraph with at least one node. 
We vary both the number $o$ of data owners and the size $n_i$ of theire subgraphs. We first fix the number of data owners at $10$, and vary the subgraph size within $\{20,40,60,80\}$ to investigate the effect of subgraph size.
Next, we fix the subgraph size at $80$, and vary the number of data owners within $\{5,10,15,20\}$ to investigate the effect of the number of data owners.  
The comparison results are presented in App.~\ref{app:subgraph parameter}. 



{\it Baselines.}
To validate the overall performance of the SIMT, we benchmark it against four baseline mechanisms. These baselines incorporate different methods for assessing data importance within our proposed model trading framework. The baselines are: 
%
%
\begin{itemize}[leftmargin=*]
    \item Greedy~\cite{singer2010budget}: The Greedy mechanism treats all data as equally important and procures data based solely on the valuations of data owners. 
    No feature propagation is applied. 
    \item ASCV~\cite{chandra2021initial}: The ASCV mechanism first trains a VGAE model on the graph to learn node embeddings with optimising reconstruction loss, and evaluates data importance by the nodes' contribution to the reconstruction loss. 
    The greater the contribution of the node, the more important the corresponding owner's data is.
    Then the auction procures data according to the ratio of data importance to valuation. 
    No feature propagation is applied. 
    \item Greedy(P): The Greedy(P) mechanism is the same as the Greedy except for that a feature imputation is applied. 
    \item ASCV(P): The ASCP(P) mechanism is the same as the ASCV except for that a feature imputation is applied.
\end{itemize}
Note that ASCV is originally designed using various techniques to evaluate data importance. However, in the absence of features, only the VGAE technique can be directly applied in our scenario. For fair comparison, 
we redesign all baselines to avoid pre-purchase inspection of data, and set same seeds for all places involving randomness, including edge augmentation and model initialisation.

{\it Implementation.}
All experiments are conducted on a single machine with AMD Ryzen $9$ $5900$X $12$-Core CPU, and NVIDIA GeForce RTX $3090$ GPU. The code is implemented using Python $3.10$ and Pytorch $2.0.1$. Our code is available in the supplementary material. 

\subsection{Overall performance} 

The experiment results are presented in Table~\ref{tab:main-performance}. Here, we fix the number of data owners at $10$, with each owner holding $80$ data subjects. 
As shown in the table, SIMT consistently outperforms all baselines under all budgets. Compared to the vanilla Greedy and ASCV, SIMT improves up to $40\%$ in both MacroF1 and MicroF1. Also, the last column shows the contribution per node, i.e., calculated as the average accuracy divided by the number of purchased nodes. The results consistently show that the contribution per node of SIMT is higher than that of all baselines, 
demonstrating the data selected by SIMT is more valuable. This validates the effectiveness of our structural importance assessment method in Sec.~\ref{sec:structural-importance}.

Table~\ref{tab:main-performance} also shows that the ASCV/ASCV(P) mechanism outperforms the Greedy/Greedy(P) mechanism. 
This could be attributed to that both ASCV and ASCV(P) assess the importance of data based on their structural contribution to the reconstruction loss, which, in a way, reflects structural uncertainty.
However, ASCV and ASCV(P) do not perform as well as SIMT. 
This discrepancy underscores the effectiveness of structural importance score.

%
Lastly, Table~\ref{tab:main-performance} shows that the ASCV(P) and Greedy(P) outperform their vanilla versions by up to $20\%$ 
in both MacroF1 and MicroF1. This validate the need of feature imputation. 

Same trend is observed in the  scenario with $n_i=1, \forall i\in O$. See more details in Table~\ref{apptab:GCN} of App.~\ref{app:models}.



\subsection{Ablation study} \label{Sec:ablation}
To explore the impact of each component in SIMT on its performance, we conduct ablation studies across the five datasets and present the average test accuracy.
As shown in Table~\ref{tab:ablation-study}, the four components, i.e., structuring clustering (clust), structural informativeness (info), structural representativeness (rep), and edge augmentation (edge aug), distinctly enhances SIMT's performance. In general, clust plays contributes the most among all components, which underscores the crucial role of structural clustering. Without clustering, there is a high probability that the procured data are unevenly distributed across the classes, leading to a biased training dataset. The second contributor is edge augmentation, which highlights the role of missing edge augmentation in the training process. Without edge augmentation, the message passing process is likely hindered, resulting in sub-optimal performance.

We also compare the effect of data valuations deviation (in App.~\ref{app:deviation}), graph centrality metrics (in App.~\ref{app: centrality}), subgraph parameters (in App.~\ref{app:subgraph parameter}), and GNN architectures (in App.~\ref{app:models}).


\begin{table}[]
\centering
\scriptsize
\caption{The impact of each component}
\label{tab:ablation-study}
\begin{tabular}{|c|cc|cc|cc|cc|cc|l}
\cline{1-11}
dataset       & \multicolumn{2}{c|}{Cora} & \multicolumn{2}{c|}{Citeseer} & \multicolumn{2}{c|}{Pubmed} & \multicolumn{2}{c|}{Amazon} & \multicolumn{2}{c|}{Coauthor} & \multicolumn{1}{c}{} \\ \cline{1-11}
metric        & acc.          & $\Delta$  & acc.            & $\Delta$    & acc.           & $\Delta$   &  acc.           & $\Delta$   & acc.            & $\Delta$    &                      \\ \cline{1-11}
no cluster    & 55.9          & -6.6      & 48.0            & -2.7        & 41.8           & -8.1       &  35.4          & -16.8       & 45.5            & -16.5        &                      \\
no rep        & 61.0          & -1.5      & 50.6            & -0.1        & 48.3           & -1.6       &  50.5          & -1.7       &  56.9           & -5.1        &                      \\
no info       & 62.3          & -0.2      & 49.5            & -1.2        & 47.8           & -2.1       &  50.9        & -1.3       &  55.1           & -6.9        &                      \\
no edge aug   & 59.2          & -3.3      & 46.2            & -4.5        & 48.7           & -1.2       &  41.1          & -11.1       &  58.3           & -3.7        &                      \\ \cline{1-11}
\textbf{SIMT} & $\bm{62.5}$   &           & $\bm{50.7}$     &             & $\bm{49.9}$    &            & $\bm{52.2}$    &            & $\bm{62.0}$     &             &                      \\ \cline{1-11}
\end{tabular}
\end{table}





\section{Conclusion}
In this paper, we aim to design a mechanism that properly incentives data owners to contribute their data, and returns a well performing GNN model to the model consumer for model marketplaces. In particular, we focus on the question of how we can measure data importance for model training without direct inspection. We propose SIMT, which consists of a data procurement phase and a model training phase. For data procurement, we incorporate a structure-based importance assessment method into an auction mechanism. For model training, we introduce and design two effective methods to impute missing data. As a result, SIMT ensures no data disclosure and incentive properties. Experimental results demonstrate that SIMT outperforms the baselines by up to $40\%$ in accuracy. 
To the best of our knowledge, SIMT is the first model trading mechanism addressing the data disclosure problem. In the future, we will further consider the potential privacy leakage in the trained model.



\begin{acks}
This work is supported by National Natural Science Foundation of China $\#62172077$ and China Scholarship Council Grant $\#201906030067$. 
\end{acks}



\newpage

\balance

\bibliographystyle{ACM-Reference-Format} 
\bibliography{sample}


\begin{thebibliography}{50}


\ifx \showCODEN    \undefined \def \showCODEN     #1{\unskip}     \fi
\ifx \showDOI      \undefined \def \showDOI       #1{#1}\fi
\ifx \showISBNx    \undefined \def \showISBNx     #1{\unskip}     \fi
\ifx \showISBNxiii \undefined \def \showISBNxiii  #1{\unskip}     \fi
\ifx \showISSN     \undefined \def \showISSN      #1{\unskip}     \fi
\ifx \showLCCN     \undefined \def \showLCCN      #1{\unskip}     \fi
\ifx \shownote     \undefined \def \shownote      #1{#1}          \fi
\ifx \showarticletitle \undefined \def \showarticletitle #1{#1}   \fi
\ifx \showURL      \undefined \def \showURL       {\relax}        \fi
\providecommand\bibfield[2]{#2}
\providecommand\bibinfo[2]{#2}
\providecommand\natexlab[1]{#1}
\providecommand\showeprint[2][]{arXiv:#2}

\bibitem[\protect\citeauthoryear{Abernethy, Chen, Ho, and Waggoner}{Abernethy
  et~al\mbox{.}}{2015}]%
        {abernethy2015low}
\bibfield{author}{\bibinfo{person}{Jacob Abernethy}, \bibinfo{person}{Yiling
  Chen}, \bibinfo{person}{Chien-Ju Ho}, {and} \bibinfo{person}{Bo Waggoner}.}
  \bibinfo{year}{2015}\natexlab{}.
\newblock \showarticletitle{Low-cost learning via active data procurement}. In
  \bibinfo{booktitle}{\emph{Proceedings of the Sixteenth ACM Conference on
  Economics and Computation}}. \bibinfo{pages}{619--636}.
\newblock


\bibitem[\protect\citeauthoryear{Agarwal, Dahleh, and Sarkar}{Agarwal
  et~al\mbox{.}}{2019}]%
        {agarwal2019marketplace}
\bibfield{author}{\bibinfo{person}{Anish Agarwal}, \bibinfo{person}{Munther
  Dahleh}, {and} \bibinfo{person}{Tuhin Sarkar}.}
  \bibinfo{year}{2019}\natexlab{}.
\newblock \showarticletitle{A marketplace for data: An algorithmic solution}.
  In \bibinfo{booktitle}{\emph{Proceedings of the 2019 ACM Conference on
  Economics and Computation}}. \bibinfo{pages}{701--726}.
\newblock


\bibitem[\protect\citeauthoryear{Bechtsis, Tsolakis, Iakovou, and
  Vlachos}{Bechtsis et~al\mbox{.}}{2022}]%
        {bechtsis2022data}
\bibfield{author}{\bibinfo{person}{Dimitrios Bechtsis}, \bibinfo{person}{Naoum
  Tsolakis}, \bibinfo{person}{Eleftherios Iakovou}, {and}
  \bibinfo{person}{Dimitrios Vlachos}.} \bibinfo{year}{2022}\natexlab{}.
\newblock \showarticletitle{Data-driven secure, resilient and sustainable
  supply chains: gaps, opportunities, and a new generalised data sharing and
  data monetisation framework}.
\newblock \bibinfo{journal}{\emph{International Journal of Production
  Research}} \bibinfo{volume}{60}, \bibinfo{number}{14} (\bibinfo{year}{2022}),
  \bibinfo{pages}{4397--4417}.
\newblock


\bibitem[\protect\citeauthoryear{Bollob{\'a}s}{Bollob{\'a}s}{1998}]%
        {bollobas1998random}
\bibfield{author}{\bibinfo{person}{B{\'e}la Bollob{\'a}s}.}
  \bibinfo{year}{1998}\natexlab{}.
\newblock \bibinfo{booktitle}{\emph{Random graphs}}.
\newblock \bibinfo{publisher}{Springer}.
\newblock


\bibitem[\protect\citeauthoryear{Brunton and Kutz}{Brunton and Kutz}{2022}]%
        {brunton2022data}
\bibfield{author}{\bibinfo{person}{Steven~L Brunton} {and}
  \bibinfo{person}{J~Nathan Kutz}.} \bibinfo{year}{2022}\natexlab{}.
\newblock \bibinfo{booktitle}{\emph{Data-driven science and engineering:
  Machine learning, dynamical systems, and control}}.
\newblock \bibinfo{publisher}{Cambridge University Press}.
\newblock


\bibitem[\protect\citeauthoryear{Burkardt}{Burkardt}{2014}]%
        {burkardt2014truncated}
\bibfield{author}{\bibinfo{person}{John Burkardt}.}
  \bibinfo{year}{2014}\natexlab{}.
\newblock \showarticletitle{The truncated normal distribution}.
\newblock \bibinfo{journal}{\emph{Department of Scientific Computing Website,
  Florida State University}}  \bibinfo{volume}{1} (\bibinfo{year}{2014}),
  \bibinfo{pages}{35}.
\newblock


\bibitem[\protect\citeauthoryear{Cai, Zheng, and Chang}{Cai
  et~al\mbox{.}}{2017}]%
        {cai2017active}
\bibfield{author}{\bibinfo{person}{Hongyun Cai}, \bibinfo{person}{Vincent~W
  Zheng}, {and} \bibinfo{person}{Kevin Chen-Chuan Chang}.}
  \bibinfo{year}{2017}\natexlab{}.
\newblock \showarticletitle{Active learning for graph embedding}.
\newblock \bibinfo{journal}{\emph{arXiv preprint arXiv:1705.05085}}
  (\bibinfo{year}{2017}).
\newblock


\bibitem[\protect\citeauthoryear{Chandra, Desai, Devaguptapu, and
  Balasubramanian}{Chandra et~al\mbox{.}}{2021}]%
        {chandra2021initial}
\bibfield{author}{\bibinfo{person}{Akshay~L Chandra},
  \bibinfo{person}{Sai~Vikas Desai}, \bibinfo{person}{Chaitanya Devaguptapu},
  {and} \bibinfo{person}{Vineeth~N Balasubramanian}.}
  \bibinfo{year}{2021}\natexlab{}.
\newblock \showarticletitle{On initial pools for deep active learning}. In
  \bibinfo{booktitle}{\emph{NeurIPS 2020 Workshop on Pre-registration in
  Machine Learning}}. PMLR, \bibinfo{pages}{14--32}.
\newblock


\bibitem[\protect\citeauthoryear{Cong, Yu, Weng, Qu, Liu, and Yiu}{Cong
  et~al\mbox{.}}{2020}]%
        {cong2020vcg}
\bibfield{author}{\bibinfo{person}{Mingshu Cong}, \bibinfo{person}{Han Yu},
  \bibinfo{person}{Xi Weng}, \bibinfo{person}{Jiabao Qu}, \bibinfo{person}{Yang
  Liu}, {and} \bibinfo{person}{Siu~Ming Yiu}.} \bibinfo{year}{2020}\natexlab{}.
\newblock \showarticletitle{A VCG-based Fair Incentive Mechanism for Federated
  Learning}.
\newblock  (\bibinfo{year}{2020}).
\newblock


\bibitem[\protect\citeauthoryear{Dandekar, Fawaz, and Ioannidis}{Dandekar
  et~al\mbox{.}}{2012}]%
        {dandekar2012privacy}
\bibfield{author}{\bibinfo{person}{Pranav Dandekar}, \bibinfo{person}{Nadia
  Fawaz}, {and} \bibinfo{person}{Stratis Ioannidis}.}
  \bibinfo{year}{2012}\natexlab{}.
\newblock \showarticletitle{Privacy auctions for recommender systems}. In
  \bibinfo{booktitle}{\emph{Proceedings of the 8th international conference on
  Internet and Network Economics}}. Springer-Verlag, \bibinfo{pages}{309--322}.
\newblock


\bibitem[\protect\citeauthoryear{Ghorbani and Zou}{Ghorbani and Zou}{2019}]%
        {ghorbani2019data}
\bibfield{author}{\bibinfo{person}{Amirata Ghorbani} {and}
  \bibinfo{person}{James Zou}.} \bibinfo{year}{2019}\natexlab{}.
\newblock \showarticletitle{Data shapley: Equitable valuation of data for
  machine learning}. In \bibinfo{booktitle}{\emph{International conference on
  machine learning}}. PMLR, \bibinfo{pages}{2242--2251}.
\newblock


\bibitem[\protect\citeauthoryear{Ghosh and Roth}{Ghosh and Roth}{2011}]%
        {ghosh2011selling}
\bibfield{author}{\bibinfo{person}{Arpita Ghosh} {and} \bibinfo{person}{Aaron
  Roth}.} \bibinfo{year}{2011}\natexlab{}.
\newblock \showarticletitle{Selling privacy at auction}. In
  \bibinfo{booktitle}{\emph{Proceedings of the 12th ACM conference on
  Electronic commerce}}. \bibinfo{pages}{199--208}.
\newblock


\bibitem[\protect\citeauthoryear{Hamilton, Ying, and Leskovec}{Hamilton
  et~al\mbox{.}}{2017}]%
        {hamilton2017inductive}
\bibfield{author}{\bibinfo{person}{Will Hamilton}, \bibinfo{person}{Zhitao
  Ying}, {and} \bibinfo{person}{Jure Leskovec}.}
  \bibinfo{year}{2017}\natexlab{}.
\newblock \showarticletitle{Inductive representation learning on large graphs}.
\newblock \bibinfo{journal}{\emph{Advances in neural information processing
  systems}}  \bibinfo{volume}{30} (\bibinfo{year}{2017}).
\newblock


\bibitem[\protect\citeauthoryear{Heckman, Boehmer, Peters, Davaloo, and
  Kurup}{Heckman et~al\mbox{.}}{2015}]%
        {heckman2015pricing}
\bibfield{author}{\bibinfo{person}{Judd~Randolph Heckman},
  \bibinfo{person}{Erin~Laurel Boehmer}, \bibinfo{person}{Elizabeth~Hope
  Peters}, \bibinfo{person}{Milad Davaloo}, {and}
  \bibinfo{person}{Nikhil~Gopinath Kurup}.} \bibinfo{year}{2015}\natexlab{}.
\newblock \showarticletitle{A pricing model for data markets}.
\newblock \bibinfo{journal}{\emph{IConference 2015 Proceedings}}
  (\bibinfo{year}{2015}).
\newblock


\bibitem[\protect\citeauthoryear{Hu and Gong}{Hu and Gong}{2020}]%
        {hu2020trading}
\bibfield{author}{\bibinfo{person}{Rui Hu} {and} \bibinfo{person}{Yanmin
  Gong}.} \bibinfo{year}{2020}\natexlab{}.
\newblock \showarticletitle{Trading data for learning: Incentive mechanism for
  on-device federated learning}. In \bibinfo{booktitle}{\emph{GLOBECOM
  2020-2020 IEEE Global Communications Conference}}. IEEE,
  \bibinfo{pages}{1--6}.
\newblock


\bibitem[\protect\citeauthoryear{Jaisingh, Barron, Mehta, and
  Chaturvedi}{Jaisingh et~al\mbox{.}}{2008}]%
        {jaisingh2008privacy}
\bibfield{author}{\bibinfo{person}{Jeevan Jaisingh}, \bibinfo{person}{Jack
  Barron}, \bibinfo{person}{Shailendra Mehta}, {and} \bibinfo{person}{Alok
  Chaturvedi}.} \bibinfo{year}{2008}\natexlab{}.
\newblock \showarticletitle{Privacy and pricing personal information}.
\newblock \bibinfo{journal}{\emph{European Journal of Operational Research}}
  \bibinfo{volume}{187}, \bibinfo{number}{3} (\bibinfo{year}{2008}),
  \bibinfo{pages}{857--870}.
\newblock


\bibitem[\protect\citeauthoryear{Jia, Dao, Wang, Hubis, Gurel, Li, Zhang,
  Spanos, and Song}{Jia et~al\mbox{.}}{2019}]%
        {jia2019efficient}
\bibfield{author}{\bibinfo{person}{Ruoxi Jia}, \bibinfo{person}{David Dao},
  \bibinfo{person}{Boxin Wang}, \bibinfo{person}{Frances~Ann Hubis},
  \bibinfo{person}{Nezihe~Merve Gurel}, \bibinfo{person}{Bo Li},
  \bibinfo{person}{Ce Zhang}, \bibinfo{person}{Costas~J. Spanos}, {and}
  \bibinfo{person}{Dawn Song}.} \bibinfo{year}{2019}\natexlab{}.
\newblock \showarticletitle{Efficient Task-Specific Data Valuation for Nearest
  Neighbor Algorithms}.
\newblock \bibinfo{journal}{\emph{Proceedings of the VLDB Endowment}}
  (\bibinfo{year}{2019}).
\newblock


\bibitem[\protect\citeauthoryear{Kipf and Welling}{Kipf and Welling}{2016a}]%
        {kipf2016semi}
\bibfield{author}{\bibinfo{person}{Thomas~N Kipf} {and} \bibinfo{person}{Max
  Welling}.} \bibinfo{year}{2016}\natexlab{a}.
\newblock \showarticletitle{Semi-supervised classification with graph
  convolutional networks}.
\newblock \bibinfo{journal}{\emph{arXiv preprint arXiv:1609.02907}}
  (\bibinfo{year}{2016}).
\newblock


\bibitem[\protect\citeauthoryear{Kipf and Welling}{Kipf and Welling}{2016b}]%
        {kipf2016variational}
\bibfield{author}{\bibinfo{person}{Thomas~N Kipf} {and} \bibinfo{person}{Max
  Welling}.} \bibinfo{year}{2016}\natexlab{b}.
\newblock \showarticletitle{Variational graph auto-encoders}.
\newblock \bibinfo{journal}{\emph{arXiv preprint arXiv:1611.07308}}
  (\bibinfo{year}{2016}).
\newblock


\bibitem[\protect\citeauthoryear{Kushal, Moorthy, and Kumar}{Kushal
  et~al\mbox{.}}{2012}]%
        {kushal2012pricing}
\bibfield{author}{\bibinfo{person}{Avanish Kushal}, \bibinfo{person}{Sharmadha
  Moorthy}, {and} \bibinfo{person}{Vikash Kumar}.}
  \bibinfo{year}{2012}\natexlab{}.
\newblock \showarticletitle{Pricing for data markets}.
\newblock \bibinfo{journal}{\emph{online] https://courses. cs. washington.
  edu/courses/cse544/11wi/projects/kumar\_kushal\_moorthy. pdf}}
  (\bibinfo{year}{2012}).
\newblock


\bibitem[\protect\citeauthoryear{Latora, Nicosia, and Russo}{Latora
  et~al\mbox{.}}{2017}]%
        {latora2017complex}
\bibfield{author}{\bibinfo{person}{Vito Latora}, \bibinfo{person}{Vincenzo
  Nicosia}, {and} \bibinfo{person}{Giovanni Russo}.}
  \bibinfo{year}{2017}\natexlab{}.
\newblock \bibinfo{booktitle}{\emph{Complex networks: principles, methods and
  applications}}.
\newblock \bibinfo{publisher}{Cambridge University Press}.
\newblock


\bibitem[\protect\citeauthoryear{Li and Pan}{Li and Pan}{2016}]%
        {li2016structural}
\bibfield{author}{\bibinfo{person}{Angsheng Li} {and} \bibinfo{person}{Yicheng
  Pan}.} \bibinfo{year}{2016}\natexlab{}.
\newblock \showarticletitle{Structural information and dynamical complexity of
  networks}.
\newblock \bibinfo{journal}{\emph{IEEE Transactions on Information Theory}}
  \bibinfo{volume}{62}, \bibinfo{number}{6} (\bibinfo{year}{2016}),
  \bibinfo{pages}{3290--3339}.
\newblock


\bibitem[\protect\citeauthoryear{Ligett and Roth}{Ligett and Roth}{2012}]%
        {ligett2012take}
\bibfield{author}{\bibinfo{person}{Katrina Ligett} {and} \bibinfo{person}{Aaron
  Roth}.} \bibinfo{year}{2012}\natexlab{}.
\newblock \showarticletitle{Take it or leave it: Running a survey when privacy
  comes at a cost}. In \bibinfo{booktitle}{\emph{International Workshop on
  Internet and Network Economics}}. Springer, \bibinfo{pages}{378--391}.
\newblock


\bibitem[\protect\citeauthoryear{Liu}{Liu}{2020}]%
        {liu2020dealer}
\bibfield{author}{\bibinfo{person}{Jinfei Liu}.}
  \bibinfo{year}{2020}\natexlab{}.
\newblock \showarticletitle{Dealer: end-to-end data marketplace with
  model-based pricing}.
\newblock \bibinfo{journal}{\emph{arXiv preprint arXiv:2003.13103}}
  (\bibinfo{year}{2020}).
\newblock


\bibitem[\protect\citeauthoryear{Liu, Liu, Zhang, Zhu, and Li}{Liu
  et~al\mbox{.}}{2019}]%
        {liu2019rem}
\bibfield{author}{\bibinfo{person}{Yiwei Liu}, \bibinfo{person}{Jiamou Liu},
  \bibinfo{person}{Zijian Zhang}, \bibinfo{person}{Liehuang Zhu}, {and}
  \bibinfo{person}{Angsheng Li}.} \bibinfo{year}{2019}\natexlab{}.
\newblock \showarticletitle{REM: From structural entropy to community structure
  deception}.
\newblock \bibinfo{journal}{\emph{Advances in Neural Information Processing
  Systems}}  \bibinfo{volume}{32} (\bibinfo{year}{2019}).
\newblock


\bibitem[\protect\citeauthoryear{Ma, Guan, and Zhao}{Ma et~al\mbox{.}}{2008}]%
        {ma2008bringing}
\bibfield{author}{\bibinfo{person}{Nan Ma}, \bibinfo{person}{Jiancheng Guan},
  {and} \bibinfo{person}{Yi Zhao}.} \bibinfo{year}{2008}\natexlab{}.
\newblock \showarticletitle{Bringing PageRank to the citation analysis}.
\newblock \bibinfo{journal}{\emph{Information Processing \& Management}}
  \bibinfo{volume}{44}, \bibinfo{number}{2} (\bibinfo{year}{2008}),
  \bibinfo{pages}{800--810}.
\newblock


\bibitem[\protect\citeauthoryear{McPherson, Smith-Lovin, and Cook}{McPherson
  et~al\mbox{.}}{2001}]%
        {mcpherson2001birds}
\bibfield{author}{\bibinfo{person}{Miller McPherson}, \bibinfo{person}{Lynn
  Smith-Lovin}, {and} \bibinfo{person}{James~M Cook}.}
  \bibinfo{year}{2001}\natexlab{}.
\newblock \showarticletitle{Birds of a feather: Homophily in social networks}.
\newblock \bibinfo{journal}{\emph{Annual review of sociology}}
  \bibinfo{volume}{27}, \bibinfo{number}{1} (\bibinfo{year}{2001}),
  \bibinfo{pages}{415--444}.
\newblock


\bibitem[\protect\citeauthoryear{Ohrimenko, Tople, and Tschiatschek}{Ohrimenko
  et~al\mbox{.}}{2019}]%
        {ohrimenko2019collaborative}
\bibfield{author}{\bibinfo{person}{Olga Ohrimenko}, \bibinfo{person}{Shruti
  Tople}, {and} \bibinfo{person}{Sebastian Tschiatschek}.}
  \bibinfo{year}{2019}\natexlab{}.
\newblock \showarticletitle{Collaborative machine learning markets with
  data-replication-robust payments}.
\newblock \bibinfo{journal}{\emph{arXiv preprint arXiv:1911.09052}}
  (\bibinfo{year}{2019}).
\newblock


\bibitem[\protect\citeauthoryear{Oord, Li, and Vinyals}{Oord
  et~al\mbox{.}}{2018}]%
        {oord2018representation}
\bibfield{author}{\bibinfo{person}{Aaron van~den Oord}, \bibinfo{person}{Yazhe
  Li}, {and} \bibinfo{person}{Oriol Vinyals}.} \bibinfo{year}{2018}\natexlab{}.
\newblock \showarticletitle{Representation learning with contrastive predictive
  coding}.
\newblock \bibinfo{journal}{\emph{arXiv preprint arXiv:1807.03748}}
  (\bibinfo{year}{2018}).
\newblock


\bibitem[\protect\citeauthoryear{Parra-Arnau}{Parra-Arnau}{2018}]%
        {parra2018optimized}
\bibfield{author}{\bibinfo{person}{Javier Parra-Arnau}.}
  \bibinfo{year}{2018}\natexlab{}.
\newblock \showarticletitle{Optimized, direct sale of privacy in personal data
  marketplaces}.
\newblock \bibinfo{journal}{\emph{Information Sciences}}  \bibinfo{volume}{424}
  (\bibinfo{year}{2018}), \bibinfo{pages}{354--384}.
\newblock


\bibitem[\protect\citeauthoryear{Ren, Xiao, Chang, Huang, Li, Gupta, Chen, and
  Wang}{Ren et~al\mbox{.}}{2021}]%
        {ren2021survey}
\bibfield{author}{\bibinfo{person}{Pengzhen Ren}, \bibinfo{person}{Yun Xiao},
  \bibinfo{person}{Xiaojun Chang}, \bibinfo{person}{Po-Yao Huang},
  \bibinfo{person}{Zhihui Li}, \bibinfo{person}{Brij~B Gupta},
  \bibinfo{person}{Xiaojiang Chen}, {and} \bibinfo{person}{Xin Wang}.}
  \bibinfo{year}{2021}\natexlab{}.
\newblock \showarticletitle{A survey of deep active learning}.
\newblock \bibinfo{journal}{\emph{ACM computing surveys (CSUR)}}
  \bibinfo{volume}{54}, \bibinfo{number}{9} (\bibinfo{year}{2021}),
  \bibinfo{pages}{1--40}.
\newblock


\bibitem[\protect\citeauthoryear{Rossi, Kenlay, Gorinova, Chamberlain, Dong,
  and Bronstein}{Rossi et~al\mbox{.}}{2022}]%
        {rossi2022unreasonable}
\bibfield{author}{\bibinfo{person}{Emanuele Rossi}, \bibinfo{person}{Henry
  Kenlay}, \bibinfo{person}{Maria~I Gorinova}, \bibinfo{person}{Benjamin~Paul
  Chamberlain}, \bibinfo{person}{Xiaowen Dong}, {and}
  \bibinfo{person}{Michael~M Bronstein}.} \bibinfo{year}{2022}\natexlab{}.
\newblock \showarticletitle{On the unreasonable effectiveness of feature
  propagation in learning on graphs with missing node features}. In
  \bibinfo{booktitle}{\emph{Learning on Graphs Conference}}. PMLR,
  \bibinfo{pages}{11--1}.
\newblock


\bibitem[\protect\citeauthoryear{Roth and Schoenebeck}{Roth and
  Schoenebeck}{2012}]%
        {roth2012conducting}
\bibfield{author}{\bibinfo{person}{Aaron Roth} {and} \bibinfo{person}{Grant
  Schoenebeck}.} \bibinfo{year}{2012}\natexlab{}.
\newblock \showarticletitle{Conducting truthful surveys, cheaply}. In
  \bibinfo{booktitle}{\emph{Proceedings of the 13th ACM Conference on
  Electronic Commerce}}. \bibinfo{pages}{826--843}.
\newblock


\bibitem[\protect\citeauthoryear{Shapley}{Shapley}{1951}]%
        {shapley1951notes}
\bibfield{author}{\bibinfo{person}{Lloyd~S Shapley}.}
  \bibinfo{year}{1951}\natexlab{}.
\newblock \showarticletitle{Notes on the n-person game—ii: The value of an
  n-person game}.
\newblock  (\bibinfo{year}{1951}).
\newblock


\bibitem[\protect\citeauthoryear{Shchur, Mumme, Bojchevski, and
  G{\"u}nnemann}{Shchur et~al\mbox{.}}{2018}]%
        {shchur2018pitfalls}
\bibfield{author}{\bibinfo{person}{Oleksandr Shchur},
  \bibinfo{person}{Maximilian Mumme}, \bibinfo{person}{Aleksandar Bojchevski},
  {and} \bibinfo{person}{Stephan G{\"u}nnemann}.}
  \bibinfo{year}{2018}\natexlab{}.
\newblock \showarticletitle{Pitfalls of graph neural network evaluation}.
\newblock \bibinfo{journal}{\emph{arXiv preprint arXiv:1811.05868}}
  (\bibinfo{year}{2018}).
\newblock


\bibitem[\protect\citeauthoryear{Sim, Zhang, Chan, and Low}{Sim
  et~al\mbox{.}}{2020}]%
        {sim2020collaborative}
\bibfield{author}{\bibinfo{person}{Rachael Hwee~Ling Sim},
  \bibinfo{person}{Yehong Zhang}, \bibinfo{person}{Mun~Choon Chan}, {and}
  \bibinfo{person}{Bryan Kian~Hsiang Low}.} \bibinfo{year}{2020}\natexlab{}.
\newblock \showarticletitle{Collaborative machine learning with incentive-aware
  model rewards}. In \bibinfo{booktitle}{\emph{International conference on
  machine learning}}. PMLR, \bibinfo{pages}{8927--8936}.
\newblock


\bibitem[\protect\citeauthoryear{Singer}{Singer}{2010}]%
        {singer2010budget}
\bibfield{author}{\bibinfo{person}{Yaron Singer}.}
  \bibinfo{year}{2010}\natexlab{}.
\newblock \showarticletitle{Budget feasible mechanisms}. In
  \bibinfo{booktitle}{\emph{2010 IEEE 51st Annual Symposium on foundations of
  computer science}}. IEEE, \bibinfo{pages}{765--774}.
\newblock


\bibitem[\protect\citeauthoryear{Spinelli, Scardapane, and Uncini}{Spinelli
  et~al\mbox{.}}{2020}]%
        {spinelli2020missing}
\bibfield{author}{\bibinfo{person}{Indro Spinelli}, \bibinfo{person}{Simone
  Scardapane}, {and} \bibinfo{person}{Aurelio Uncini}.}
  \bibinfo{year}{2020}\natexlab{}.
\newblock \showarticletitle{Missing data imputation with adversarially-trained
  graph convolutional networks}.
\newblock \bibinfo{journal}{\emph{Neural Networks}}  \bibinfo{volume}{129}
  (\bibinfo{year}{2020}), \bibinfo{pages}{249--260}.
\newblock


\bibitem[\protect\citeauthoryear{Sun, Chen, Liao, and Huang}{Sun
  et~al\mbox{.}}{2022}]%
        {sun2022profit}
\bibfield{author}{\bibinfo{person}{Peng Sun}, \bibinfo{person}{Xu Chen},
  \bibinfo{person}{Guocheng Liao}, {and} \bibinfo{person}{Jianwei Huang}.}
  \bibinfo{year}{2022}\natexlab{}.
\newblock \showarticletitle{A profit-maximizing model marketplace with
  differentially private federated learning}. In \bibinfo{booktitle}{\emph{IEEE
  INFOCOM 2022-IEEE Conference on Computer Communications}}. IEEE,
  \bibinfo{pages}{1439--1448}.
\newblock


\bibitem[\protect\citeauthoryear{Taguchi, Liu, and Murata}{Taguchi
  et~al\mbox{.}}{2021}]%
        {taguchi2021graph}
\bibfield{author}{\bibinfo{person}{Hibiki Taguchi}, \bibinfo{person}{Xin Liu},
  {and} \bibinfo{person}{Tsuyoshi Murata}.} \bibinfo{year}{2021}\natexlab{}.
\newblock \showarticletitle{Graph convolutional networks for graphs containing
  missing features}.
\newblock \bibinfo{journal}{\emph{Future Generation Computer Systems}}
  \bibinfo{volume}{117} (\bibinfo{year}{2021}), \bibinfo{pages}{155--168}.
\newblock


\bibitem[\protect\citeauthoryear{Tang, Chen, Wang, Xuan, and Zhao}{Tang
  et~al\mbox{.}}{2023}]%
        {tang2023generalized}
\bibfield{author}{\bibinfo{person}{Bisheng Tang}, \bibinfo{person}{Xiaojun
  Chen}, \bibinfo{person}{Shaopu Wang}, \bibinfo{person}{Yuexin Xuan}, {and}
  \bibinfo{person}{Zhendong Zhao}.} \bibinfo{year}{2023}\natexlab{}.
\newblock \showarticletitle{Generalized heterophily graph data augmentation for
  node classification}.
\newblock \bibinfo{journal}{\emph{Neural Networks}} (\bibinfo{year}{2023}).
\newblock
\urldef\tempurl%
\url{https://doi.org/10.1016/j.neunet.2023.09.021}
\showDOI{\tempurl}


\bibitem[\protect\citeauthoryear{Veli{\v{c}}kovi{\'c}, Cucurull, Casanova,
  Romero, Lio, and Bengio}{Veli{\v{c}}kovi{\'c} et~al\mbox{.}}{2017}]%
        {velivckovic2017graph}
\bibfield{author}{\bibinfo{person}{Petar Veli{\v{c}}kovi{\'c}},
  \bibinfo{person}{Guillem Cucurull}, \bibinfo{person}{Arantxa Casanova},
  \bibinfo{person}{Adriana Romero}, \bibinfo{person}{Pietro Lio}, {and}
  \bibinfo{person}{Yoshua Bengio}.} \bibinfo{year}{2017}\natexlab{}.
\newblock \showarticletitle{Graph attention networks}.
\newblock \bibinfo{journal}{\emph{arXiv preprint arXiv:1710.10903}}
  (\bibinfo{year}{2017}).
\newblock


\bibitem[\protect\citeauthoryear{Wang, Wang, Zhang, Yang, Zhao, and Liu}{Wang
  et~al\mbox{.}}{2023}]%
        {wang2023user}
\bibfield{author}{\bibinfo{person}{Yifei Wang}, \bibinfo{person}{Yupan Wang},
  \bibinfo{person}{Zeyu Zhang}, \bibinfo{person}{Song Yang},
  \bibinfo{person}{Kaiqi Zhao}, {and} \bibinfo{person}{Jiamou Liu}.}
  \bibinfo{year}{2023}\natexlab{}.
\newblock \showarticletitle{User: Unsupervised structural entropy-based robust
  graph neural network}.
\newblock \bibinfo{journal}{\emph{arXiv preprint arXiv:2302.05889}}
  (\bibinfo{year}{2023}).
\newblock


\bibitem[\protect\citeauthoryear{Xu, Hu, Leskovec, and Jegelka}{Xu
  et~al\mbox{.}}{2018}]%
        {xu2018powerful}
\bibfield{author}{\bibinfo{person}{Keyulu Xu}, \bibinfo{person}{Weihua Hu},
  \bibinfo{person}{Jure Leskovec}, {and} \bibinfo{person}{Stefanie Jegelka}.}
  \bibinfo{year}{2018}\natexlab{}.
\newblock \showarticletitle{How powerful are graph neural networks?}
\newblock \bibinfo{journal}{\emph{arXiv preprint arXiv:1810.00826}}
  (\bibinfo{year}{2018}).
\newblock


\bibitem[\protect\citeauthoryear{Xu, Jiang, Chen, Ren, and Liu}{Xu
  et~al\mbox{.}}{2015}]%
        {xu2015privacy}
\bibfield{author}{\bibinfo{person}{Lei Xu}, \bibinfo{person}{Chunxiao Jiang},
  \bibinfo{person}{Yan Chen}, \bibinfo{person}{Yong Ren}, {and}
  \bibinfo{person}{K.~J.~Ray Liu}.} \bibinfo{year}{2015}\natexlab{}.
\newblock \showarticletitle{Privacy or Utility in Data Collection? A Contract
  Theoretic Approach}.
\newblock \bibinfo{journal}{\emph{IEEE Journal of Selected Topics in Signal
  Processing}} \bibinfo{volume}{9}, \bibinfo{number}{7} (\bibinfo{year}{2015}),
  \bibinfo{pages}{1256--1269}.
\newblock


\bibitem[\protect\citeauthoryear{Xu, Lyu, Ma, Miao, Foo, and Low}{Xu
  et~al\mbox{.}}{2021}]%
        {xu2021gradient}
\bibfield{author}{\bibinfo{person}{Xinyi Xu}, \bibinfo{person}{Lingjuan Lyu},
  \bibinfo{person}{Xingjun Ma}, \bibinfo{person}{Chenglin Miao},
  \bibinfo{person}{Chuan~Sheng Foo}, {and} \bibinfo{person}{Bryan Kian~Hsiang
  Low}.} \bibinfo{year}{2021}\natexlab{}.
\newblock \showarticletitle{Gradient driven rewards to guarantee fairness in
  collaborative machine learning}.
\newblock \bibinfo{journal}{\emph{Advances in Neural Information Processing
  Systems}}  \bibinfo{volume}{34} (\bibinfo{year}{2021}),
  \bibinfo{pages}{16104--16117}.
\newblock


\bibitem[\protect\citeauthoryear{Zhang, Xie, Bai, Yu, Li, and Gao}{Zhang
  et~al\mbox{.}}{2021c}]%
        {zhang2021survey}
\bibfield{author}{\bibinfo{person}{Chen Zhang}, \bibinfo{person}{Yu Xie},
  \bibinfo{person}{Hang Bai}, \bibinfo{person}{Bin Yu},
  \bibinfo{person}{Weihong Li}, {and} \bibinfo{person}{Yuan Gao}.}
  \bibinfo{year}{2021}\natexlab{c}.
\newblock \showarticletitle{A survey on federated learning}.
\newblock \bibinfo{journal}{\emph{Knowledge-Based Systems}}
  \bibinfo{volume}{216} (\bibinfo{year}{2021}), \bibinfo{pages}{106775}.
\newblock


\bibitem[\protect\citeauthoryear{Zhang, Wu, and Pan}{Zhang
  et~al\mbox{.}}{2021b}]%
        {zhang2021incentive}
\bibfield{author}{\bibinfo{person}{Jingwen Zhang}, \bibinfo{person}{Yuezhou
  Wu}, {and} \bibinfo{person}{Rong Pan}.} \bibinfo{year}{2021}\natexlab{b}.
\newblock \showarticletitle{Incentive mechanism for horizontal federated
  learning based on reputation and reverse auction}. In
  \bibinfo{booktitle}{\emph{Proceedings of the Web Conference 2021}}.
  \bibinfo{pages}{947--956}.
\newblock


\bibitem[\protect\citeauthoryear{Zhang, Beltran, and Liu}{Zhang
  et~al\mbox{.}}{2020}]%
        {zhang2020selling}
\bibfield{author}{\bibinfo{person}{Mengxiao Zhang}, \bibinfo{person}{Fernando
  Beltran}, {and} \bibinfo{person}{Jiamou Liu}.}
  \bibinfo{year}{2020}\natexlab{}.
\newblock \showarticletitle{Selling Data at an Auction under Privacy
  Constraints}. In \bibinfo{booktitle}{\emph{Conference on Uncertainty in
  Artificial Intelligence}}. PMLR, \bibinfo{pages}{669--678}.
\newblock


\bibitem[\protect\citeauthoryear{Zhang, Shen, Li, Chen, Yang, and Cui}{Zhang
  et~al\mbox{.}}{2021a}]%
        {zhang2021alg}
\bibfield{author}{\bibinfo{person}{Wentao Zhang}, \bibinfo{person}{Yu Shen},
  \bibinfo{person}{Yang Li}, \bibinfo{person}{Lei Chen}, \bibinfo{person}{Zhi
  Yang}, {and} \bibinfo{person}{Bin Cui}.} \bibinfo{year}{2021}\natexlab{a}.
\newblock \showarticletitle{Alg: Fast and accurate active learning framework
  for graph convolutional networks}. In \bibinfo{booktitle}{\emph{Proceedings
  of the 2021 International Conference on Management of Data}}.
  \bibinfo{pages}{2366--2374}.
\newblock


\end{thebibliography}

\clearpage

\appendix

\noindent {\bf \Large APPENDIX}

\section{Normalised marginal structural entropy} \label{app:calculation}
\noindent The calculation of the normalised marginal structural entropy is as follows. The marginal structural entropy is: 
{\footnotesize
\begin{align*} 
    &\mathcal{H}_P(G) - \mathcal{H}_{P'}(G)  \\
    =& \frac{d_t-g_t}{2|E|}\log\frac{d_t}{2|E|}-\frac{d_t - g_t - 2x_{v,t}}{2|E|}\log \frac{d_t-d_v}{2|E|} \\
    =& \frac{d_t - g_t}{ 2|E|} \left(\log \frac{d_t}{2|E|} -  \log \frac{d_t-d_v}{2|E|}\right) + \frac{2x_{v,t}}{2|E|} \log \frac{d_t-d_v}{2|E|} \\
    =& \frac{d_t - g_t}{ 2|E|} \log \frac{d_t}{d_t-d_v} + \frac{2x_{v,t}}{2|E|} \log \frac{d_t-d_v}{2|E|}.
\end{align*}
}
\noindent We normalise this value and get 
\begin{align*}
    \epsilon_v & = (\mathcal{H}_P(G) - \mathcal{H}_{P'}(G) )/ \mathcal{H}_P(G) \\
    & = \frac{\frac{d_t - g_t}{ 2|E|} \log \frac{d_t}{d_t-d_v} + \frac{2x_{v,t}}{2|E|} \log \frac{d_t-d_v}{2|E|} }{ \frac{d_t-g_t}{2|E|}\log\frac{d_t}{2|E|} } \\
    & = \frac{(d_t-g_t)\log\frac{d_t}{d_t-d_v}+2x_{v,t}\log\frac{d_t-d_v}{2|E|}}{(d_t-g_t)\log\frac{d_t}{2|E|}}.
\end{align*}
\section{Dataset statistics} \label{app: dataset}
\begin{table}[h]
\centering
\caption{Dataset statistics}
\label{tab:dataset_statistics}
\footnotesize
\begin{tabular}{|c|cccc|}
\hline
Dataset & $\#$ nodes & $\#$ edges & $\#$ features & $\#$ classes \\ \hline
Cora & 2,708 & 10,556 & 1,433 & 7 \\ \hline
Citeseer & 3,327 & 9,104 & 3,703 & 6 \\ \hline
Pubmed &19,717 & 88,648 & 500 & 3 \\ \hline
Amazon &  7,650 & 238,162 & 745 & 8 \\ \hline
Coauthor &  34,493 & 495,924 & 8,415 & 5 \\ \hline
\end{tabular}%
\end{table}

\section{Comparison of data valuations deviation} \label{app:deviation}
We conduct an experiment to evaluate the effect of the standard deviation in data valuation distribution on scenario with $n_i=1, i\in O$. 
The experiment is evaluated on Cora, Citeseer and Pubmed datasets with standard deviation $\sigma$ ranging from $0.01$ to $0.2$.
In Table~\ref{tab:deviation}, we list the MacroF1 and MicroF1 performance of our mechanism running on varying standard deviations of data valuations. We also list the number of bought data to compare the effect of $\sigma$ on data quantity.
We observe that the number of bought data increases as $\sigma$ increases. The increasing data quantity brings in more information for training GNN model and thus improves the model performance. 
To balance the effect of data quantity, we fix $\sigma=0.1$ for all the other experiments. To eliminate the randomness induced by unknown subgraphs, we conduct experiments on the case where all data owners are individuals.

\begin{table}[h]
\centering
\caption{The impact of different standard deviations for $n_i=1$}
\label{tab:deviation}
\resizebox{\linewidth}{!}{%
\begin{tabular}{c|ccc|ccc|ccc|}
\cline{2-10}
 & \multicolumn{3}{c|}{Cora} & \multicolumn{3}{c|}{Citeseer} & \multicolumn{3}{c|}{Pubmed} \\ \hline
\multicolumn{1}{|c|}{$\sigma$} & {\small MacroF1} & {\small MicroF1} & $\#$data & {\small MacroF1} & {\small MicroF1} & $\#$data & { \small MacroF1} & {\small MicroF1} & $\#$data \\ \hline
\multicolumn{1}{|c|}{0.01} & 55.2 & 64.4 & 171.1 & 46.8 & 53.6 & 161.2 & 35.3 & 47.4 & 161.2 \\
\multicolumn{1}{|c|}{0.05} & 55.8 & 65.2 & 173.3 & 49.0 & 56.0 & 162.7 & 37.7 & 48.9 & 182.0 \\
\multicolumn{1}{|c|}{0.1} & 59.2 & 68.3 & 179.4 & 50.4 & 57.0 & 168.2 & 41.8 & 52.5 & 196.7 \\
\multicolumn{1}{|c|}{0.15} & 60.4 & 69.4 & 189.0 & 52.0 & 58.7 & 175.5 & 46.2 & 56.1 & 220.2 \\
\multicolumn{1}{|c|}{0.2} & 62.6 & 70.9 & 201.7 & 54.3 & 61.0 & 185.2 & 52.4 & 60.9 & 254.1 \\ \hline
\end{tabular}%
}
\end{table}

\section{Comparison of graph centrality metrics} \label{app: centrality}
We evaluate the effectiveness of various graph centrality metrics on Cora dataset. To eliminate the randomness induced by subgraphs, we conduct experiments on the scenario with $n_i=1$ for all $i\in O$.  
Specifically, we substitute the PageRank centrality used in our structural representativeness with three other widely used centrality metrics: degree centrality, closeness centrality, and betweenness centrality \cite{latora2017complex}. We report the test accuracy with varying budgets. 
As shown in Table~\ref{tab:centrality}, PageRank centrality consistently surpasses the other metrics. Consequently, our mechanism employs PageRank centrality as structural representativeness.
\begin{table}[ht!]
\centering
\caption{Comparison of graph centrality metrics ($n_i=1$)} 
\label{tab:centrality}
\footnotesize
\begin{tabular}{|c|cccccc|}
\hline
budget & $50$ & $100$ & $150$ & $200$ & $250$ & $300$ \\ \hline
Degree & 51.8 & 60.3 & 69.2 & 69.2 & 74.2 & 75.0 \\
Closeness & 52.5 & 62.0 & 65.5 & 70.1 & 73.6 & 74.8 \\
Betweenness & 49.2 & 60.9 & 67.0 & 68.5 & 71.7 & 75.6 \\
PageRank & $\bm{52.7}$ & $\bm{64.6}$ & $\bm{69.3}$ & $\bm{71.0}$ & $\bm{75.1}$ & $\bm{77.0}$ \\ \hline
\end{tabular}%
\end{table}

\section{Comparison of subgraph parameters} \label{app:subgraph parameter}
Here we explore the effect of the subgraph size and the number of data owners.  
To reduce the effect of randomness, we use $n_i = 1$ as a specific scenario to measure the effect of other parameters.
The experimental results on the effects due to subgraph size and number of subgraphs are shown in Tables~\ref{tab:subgraph size} and \ref{tab:subgraph number}.
When fixing the budget as $1500$ and the number of data owners as $10$, we find that the larger the subgraph size is, the more edges within the subgraph are missing, which hinders the message passing process and results in worse model performance, as shown in Table~\ref{tab:subgraph size}. 
Then when fixing the budget at $1500$ and the subgraph size at $80$, Table~\ref{tab:subgraph number} shows that the more subgraphs exist, the more edges within subgraphs are missing, resulting in worse model performance.

These experiments underscore the negative effect of missing edges during model training. Consequently, we utilise an edge augmentation technique to alleviate the impact of these missing edges.

\begin{table}[ht]
\centering
\footnotesize
\caption{The  impact of subgraph size ($n_i\geq 1$)}
\label{tab:subgraph size}
\begin{tabular}{|c|c|c|c|}
\hline
size & Cora & Citeseer & Pubmed \\ \hline
20 & $63.1\pm 5.6$ & $55.4\pm 7.5$ & $52.1\pm 15.2$ \\ \hline
40 & $63.4\pm 4.2$ & $56.6\pm 3.6$ & $50.3\pm 11.7$ \\ \hline
80 & $62.5\pm 8.6$ & $50.2\pm 6.3$ & $49.2\pm 11.9$ \\ \hline
160 & $57.3\pm 4.9$ & $46.5\pm 7.0$ & $47.6\pm 8.5$ \\ \hline
\end{tabular}%
\end{table}

\begin{table}[ht]
\centering
\footnotesize
\caption{The  impact of subgraph numbers ($n_i\geq 1$)}
\label{tab:subgraph number}
\begin{tabular}{|c|c|c|c|}
\hline
\#subgraphs & Cora & Citeseer & Pubmed \\ \hline
5 & $65.5\pm 4.3$ & $53.2\pm 7.6$ & $49.7\pm 12.3$ \\ \hline
10 & $62.5\pm 8.6$ & $50.2\pm 6.3$ & $49.2\pm 11.9$ \\ \hline
15 & $62\pm 4.7$ & $50.2\pm 3.8$ & $48.6\pm 14.6$ \\ \hline
20 & $57.3\pm 5.0$ & $47.9\pm 5.0$ & $47.2\pm 11.4$ \\ \hline
\end{tabular}%
\end{table}

\begin{table}[]
\centering
\footnotesize
\caption{The number of bought data under different budgets for $n_i=1$}
\label{apptab:quantity}
\begin{tabular}{|c|c|cccccc|}
\hline
 \multicolumn{2}{|c|}{budget} & {$50$} & {$100$} & {$150$} & {$200$} & {$250$} & $300$ \\ \hline
\multirow{3}{*}{Cora} & Greedy(P) & 67.8 & 129.5 & 188.2 & 245 & 300.3 & 353.8 \\
 & ASCV(P) & 60.7 & 114 & 164.3 & 212.3 & 257.8 & 302 \\
 & {\bf SIMT} & $\bm {55.3}$ & $\bm {107.4}$ & $\bm {157}$ & $\bm {205.9}$ & $\bm {252.3}$ & $\bm {298.3}$ \\ \hline
\multirow{3}{*}{Citeseer} & Greedy(P) & 65.3 & 123.8 & 179.9 & 233.8 & 286.3 & 337.2 \\
 & ASCV(P) & 57.4 & 107 & 154.2 & 198.5 & 241.9 & 282.2 \\
 & {\bf SIMT} & $\bm {51.8}$ & $\bm {100.6}$ & $\bm {147.8}$ & $\bm {193.2}$ & $\bm {236.2}$ & $\bm {279.8}$ \\ \hline
\multirow{3}{*}{Pubmed} & Greedy(P) & 76.5 & 147.1 & 215.4 & 282 & 347.4 & 412.1 \\
 & ASCV(P) & 70.1 & 134.3 & 196 & 255.6 & 314.5 & 370.7 \\
 & {\bf SIMT} & $\bm {60.4}$ & $\bm {116.6}$ & $\bm {171.8}$ & $\bm {225}$ & $\bm {277.7}$ & $\bm {328.5}$ \\ \hline
\end{tabular}%
\end{table}

\begin{table}[ht]
\centering
\caption{The impact of each component ($n_i= 1$)}
\label{apptab:ablationstudy}
\footnotesize
\begin{tabular}{|c|cc|cc|cc|}
\hline
 dataset & \multicolumn{2}{c|}{Cora} & \multicolumn{2}{c|}{Pubmed} & \multicolumn{2}{c|}{Citeseer} \\ \hline
 metric & acc. & $\Delta$ & acc. & $\Delta$ & acc. & $\Delta$ \\ \hline
no cluster & 60.0 & -8.3 & 53.3 & -3.7 & 38.8 & -13.7 \\
no rep & 66.5 & -1.8 & 55.5 & -1.5 & 50.3 & -2.2 \\
no info & 66.5 & -1.8 & 54.9 & -2.1 & 49.6 & -2.9 \\ \hline
{\bf SIMT} & $\bm{68.3}$ &  & $\bm{57.0}$ &  & $\bm{52.5}$ &  \\ \hline
\end{tabular}%
\end{table}

\begin{table*}[t]
\centering
\caption{Node classification performance under different budgets using GCN architecture for $n_i=1$}
\label{apptab:GCN}
\resizebox{\linewidth}{!}{%
\begin{tabular}{|cc|cc|cc|cc|cc|cc|cc|}
\hline
\multicolumn{2}{|c|}{budget} & \multicolumn{2}{c|}{$50$ } & \multicolumn{2}{c|}{$100$} & \multicolumn{2}{c|}{$150$} & \multicolumn{2}{c|}{$200$} & \multicolumn{2}{c|}{$250$} & \multicolumn{2}{c|}{$300$} \\ \hline
\multicolumn{2}{|c|}{metric} & \multicolumn{1}{c}{MacroF1} & \multicolumn{1}{c|}{MicroF1} & \multicolumn{1}{c}{MacroF1} & \multicolumn{1}{c|}{MicroF1} & \multicolumn{1}{c}{MacroF1} & \multicolumn{1}{c|}{MicroF1} & \multicolumn{1}{c}{MacroF1} & \multicolumn{1}{c|}{MicroF1} & \multicolumn{1}{c}{MacroF1} & \multicolumn{1}{c|}{MicroF1} & \multicolumn{1}{c}{MacroF1} & MicroF1 \\ \hline
\multicolumn{1}{|c|}{\multirow{5}{*}{\small Cora}} & Greedy & $11.8\pm 5.5$ & $22.6\pm 8$ & $18.4\pm 7$ & $28.8\pm 9$ & $21.7\pm 6.6$ & $32.7\pm 9.4$ & $25.2\pm 5.6$ & $36.2\pm 6.3$ & $28.5\pm 8.8$ & $40.2\pm 8.5$ & $30.5\pm 9.4$ & $42.5\pm 9.4$ \\
\multicolumn{1}{|c|}{} & ASCV & $11.8\pm 5.4$ & $22.0\pm 7.8$ & $19.9\pm 5.5$ & $32.2\pm 9.2$ & $28.0\pm 6.5$ & $39.6\pm 7.7$ & $32.0\pm 4.6$ & $43.1\pm 4.3$ & $31.5\pm 3.5$ & $43.2\pm 3.8$ & $36.7\pm 5.4$ & $48.1\pm 4.3$ \\
\multicolumn{1}{|c|}{} & Greedy(P) & $26.1\pm 13.3$ & $38.2\pm 13.7$ & $31.3\pm 13$ & $42.3\pm 13.3$ & $34.1\pm 14.6$ & $44.4\pm 13.6$ & $36.8\pm 17.3$ & $47.7\pm 14$ & $38.6\pm 17.5$ & $49.1\pm 14.1$ & $40.1\pm 17$ & $51.0\pm 13.4$ \\
\multicolumn{1}{|c|}{} & ASCV(P) & $31.9\pm 8.3$ & $43.8\pm 11.1$ & $41.6\pm 8$ & $54.3\pm 8.5$ & $48.6\pm 7.7$ & $59.0\pm 7.9$ & $52.2\pm 10.7$ & $63.0\pm 8.8$ & $57.8\pm 9.6$ & $67\pm 7.6$ & $61\pm 12.6$ & $69.5\pm 8.2$ \\
\multicolumn{1}{|c|}{} & {\bf SIMT} & $\bm{39.6\pm 9.9}$ & $\bm{52.7\pm 8.3}$ & $\bm{53.4\pm 7.5}$ & $\bm{64.6\pm 7.7}$ & $\bm{60.0\pm 7.7}$ & $\bm{69.3\pm 6.3}$ & $\bm{63.4\pm 8.8}$ & $\bm{71.0\pm 6.6}$ & $\bm{67.1\pm 7.9}$ & $\bm{75.1\pm 4.3}$ & $\bm{71.9\pm 6.3}$ & $\bm{77.0\pm 4}$ \\ \hline
\multicolumn{1}{|c|}{\multirow{5}{*}{\small Citeseer}} & Greedy & $8.2\pm 4.3$ & $18.6\pm 8.9$ & $11.3\pm 4.9$ & $22.7\pm 8.7$ & $13.5\pm 4.2$ & $24.1\pm 6.8$ & $17.1\pm 5.9$ & $28.2\pm 8$ & $18.4\pm 6.4$ & $28.7\pm 8.2$ & $22.0\pm 7.2$ & $32.2\pm 7.1$ \\
\multicolumn{1}{|c|}{} & ASCV & $10.2\pm 4.4$ & $20.1\pm 5$ & $15\pm 3.8$ & $24\pm 6.2$ & $18.9\pm 3.1$ & $27.2\pm 4.5$ & $24.3\pm 5.6$ & $32.4\pm 7.9$ & $31.2\pm 5.3$ & $38\pm 6.7$ & $36.3\pm 4.5$ & $45.0\pm 4$ \\
\multicolumn{1}{|c|}{} & Greedy(P) & $17\pm 10$ & $27.6\pm 12.5$ & $22.1\pm 10.5$ & $32.6\pm 11.9$ & $25.3\pm 12.5$ & $35.2\pm 12$ & $31.3\pm 14.1$ & $40.4\pm 12.9$ & $33.2\pm 13.5$ & $42.8\pm 12.7$ & $37.6\pm 14.8$ & $47.2\pm 13.5$ \\
\multicolumn{1}{|c|}{} & ASCV(P) & $22.8\pm 9.5$ & $31.1\pm 13$ & $34.2\pm 9.8$ & $41.8\pm 12.9$ & $40.7\pm 7.3$ & $48.8\pm 8.7$ & $49.3\pm 7.4$ & $57.7\pm 6.2$ & $53.2\pm 6.9$ & $59.2\pm 5.6$ & $56.7\pm 6.1$ & $63.1\pm 3.9$ \\
\multicolumn{1}{|c|}{} & {\bf SIMT} & $\bm{33.3\pm 8.6}$ & $\bm{40.3\pm 10.5}$ & $\bm{44.2\pm 3.9}$ & $\bm{51.0\pm 7.2}$ & $\bm{50.9\pm 4.4}$ & $\bm{57.3\pm 6}$ & $\bm{55.0\pm 3.1}$ & $\bm{62.5\pm 3.5}$ & $\bm{57.9\pm 2.1}$ & $\bm{64.2\pm 3.7}$ & $\bm{61.0\pm 2.4}$ & $\bm{66.7\pm 2.3}$ \\ \hline
\multicolumn{1}{|c|}{\multirow{5}{*}{\small Pubmed}} & Greedy & $13.7\pm 3.4$ & $26.4\pm 8.7$ & $13.8\pm 3.4$ & $26.4\pm 8.7$ & $13.9\pm 3.3$ & $26.5\pm 8.6$ & $14.2\pm 3.5$ & $26.7\pm 8.5$ & $14.1\pm 3.4$ & $26.6\pm 8.6$ & $15\pm 4.7$ & $27.3\pm 8.6$ \\
\multicolumn{1}{|c|}{} & ASCV & $16.6\pm 4.5$ & $32.2\pm 9.3$ & $16.3\pm 4$ & $32.3\pm 9.4$ & $16.1\pm 3.8$ & $32.3\pm 9.4$ & $16.2\pm 3.8$ & $32.3\pm 9.4$ & $16.4\pm 3.7$ & $32.4\pm 9.3$ & $19.7\pm 7.9$ & $35.2\pm 9.6$ \\
\multicolumn{1}{|c|}{} & Greedy(P) & $16.3\pm 7.5$ & $28.5\pm 9.6$ & $17.6\pm 9.6$ & $30.0\pm 11.7$ & $19.1\pm 10.1$ & $30.9\pm 11.6$ & $20.7\pm 10.8$ & $32.1\pm 11.6$ & $22.7\pm 12$ & $33.9\pm 12$ & $24.2\pm 14.2$ & $35.2\pm 13.6$ \\
\multicolumn{1}{|c|}{} & ASCV(P) & $17.1\pm 4.5$ & $32.7\pm 9.7$ & $22.1\pm 10.2$ & $36.6\pm 11.1$ & $23.9\pm 12.3$ & $37.8\pm 12.6$ & $28.1\pm 13.4$ & $41.7\pm 12.9$ & $30.6\pm 13.2$ & $43.2\pm 13.7$ & $32.0\pm 13.3$ & $43.9\pm 13.4$ \\
\multicolumn{1}{|c|}{} & {\bf SIMT} & $\bm{29\pm 13.7}$ & $\bm{42\pm 15.9}$ & $\bm{36.9\pm 18.4}$ & $\bm{48.9\pm 17.1}$ & $\bm{40.2\pm 14.7}$ & $\bm{51.8\pm 14.5}$ & $\bm{44.6\pm 19.1}$ & $\bm{54.3\pm 17.8}$ & $\bm{47.1\pm 17.3}$ & $\bm{56.8\pm 15.8}$ & $\bm{52.8\pm 15}$ & $\bm{61.2\pm 12.2}$ \\ \hline
\end{tabular}%
}
\end{table*}

\begin{table*}[ht!]
\centering
\caption{Node classification performance under different budgets using GIN architecture for $n_i=1$}
\label{tab:GIN}
\resizebox{\linewidth}{!}{%
\begin{tabular}{|cc|cc|cc|cc|cc|cc|cc|}
\hline
\multicolumn{2}{|c|}{budget} & \multicolumn{2}{c|}{50} & \multicolumn{2}{c|}{100} & \multicolumn{2}{c|}{150} & \multicolumn{2}{c|}{200} & \multicolumn{2}{c|}{250} & \multicolumn{2}{c|}{300} \\ \hline
\multicolumn{2}{|c|}{metric} & MacroF1 & MicroF1 & MacroF1 & MicroF1 & MacroF1 & MicroF1 & MacroF1 & MicroF1 & MacroF1 & MicroF1 & MacroF1 & MicroF1 \\ \hline
\multicolumn{1}{|c|}{\multirow{5}{*}{Cora}} & Greedy & $20.5\pm 6.1$ & $30.8\pm 8.1$ & $28.4\pm 7.4$ & $38.5\pm 6.8$ & $31.2\pm 9.1$ & $40.8\pm 7.8$ & $33.5\pm 12.2$ & $43.7\pm 9.5$ & $36.8\pm 11.1$ & $45.9\pm 8.5$ & $39.1\pm 11.2$ & $48.8\pm 7.7$ \\
\multicolumn{1}{|c|}{} & ASCV & $22.8\pm 8.9$ & $29.6\pm 6.9$ & $36.0\pm 8.5$ & $40.4\pm 5.8$ & $46.2\pm 7.9$ & $50.0\pm 7.1$ & $49.9\pm 8.9$ & $53.0\pm 6.9$ & $56.9\pm 7.0$ & $59.4\pm 5.0$ & $60.8\pm 7.7$ & $62.7\pm 5.5$ \\
\multicolumn{1}{|c|}{} & Greedy(P) & $29.3\pm 13.5$ & $41.6\pm 12.1$ & $35.2\pm 15.5$ & $46.7\pm 12.4$ & $35.9\pm 15.9$ & $47.4\pm 13.3$ & $40.3\pm 15.3$ & $50.5\pm 12.8$ & $42.8\pm 14.7$ & $52.3\pm 12.2$ & $46.0\pm 14.4$ & $55.2\pm 11.0$ \\
\multicolumn{1}{|c|}{} & ASCV(P) & $35.9\pm 11.8$ & $45.7\pm 10.5$ & $49.7\pm 9.0$ & $56.4\pm 7.1$ & $55.9\pm 6.5$ & $61.3\pm 3.9$ & $68.2\pm 5.6$ & $71.2\pm 3.9$ & $71.1\pm 5.7$ & $73.7\pm 4.0$ & $74.5\pm 3.9$ & $76.1\pm 3.9$ \\
\multicolumn{1}{|c|}{} & {\bf SIMT} & $\bm{47.1\pm 10}$ & $\bm{57.4\pm 9.9}$ & $\bm{58.1\pm 7.8}$ & $\bm{65.7\pm 5.5}$ & $\bm{66.8\pm 7.7}$ & $\bm{73.1\pm 4.3}$ & $\bm{72.8\pm 5.4}$ & $\bm{76.5\pm 4.6}$ & $\bm{72.7\pm 7.1}$ & $\bm{77.6\pm 4.0}$ & $\bm{75\pm 3.9}$ & $\bm{78.2\pm 3.6}$ \\ \hline
\multicolumn{1}{|c|}{\multirow{5}{*}{Citeseer}} & Greedy & $10.2\pm 5.3$ & $18.1\pm 9.5$ & $16.1\pm 7.9$ & $24.4\pm 10.1$ & $19.8\pm 8.0$ & $26.4\pm 8.8$ & $26.7\pm 9.8$ & $33.5\pm 12.1$ & $30.4\pm 9.4$ & $37.4\pm 11.9$ & $31.3\pm 9.6$ & $37.5\pm 10.2$ \\
\multicolumn{1}{|c|}{} & ASCV & $13.1\pm 2.9$ & $20.5\pm 6.3$ & $20.4\pm 4.6$ & $24.5\pm 8.8$ & $32.4\pm 4.2$ & $34.1\pm 8.3$ & $35.8\pm 5.4$ & $35.1\pm 6.5$ & $43.0\pm 4.9$ & $45.0\pm 7.2$ & $47.4\pm 3.9$ & $50.4\pm 5.0$ \\
\multicolumn{1}{|c|}{} & Greedy(P) & $22.1\pm 12.2$ & $32.1\pm 13.7$ & $26.9\pm 11.6$ & $36.1\pm 11.7$ & $30.9\pm 12.3$ & $39.6\pm 11$ & $36.1\pm 11.7$ & $44.3\pm 11.2$ & $38.1\pm 11.6$ & $45.8\pm 11.3$ & $41.4\pm 11.8$ & $49.1\pm 11.3$ \\
\multicolumn{1}{|c|}{} & ASCV(P) & $27.6\pm 9.7$ & $35.8\pm 12.8$ & $38.5\pm 10.1$ & $45.4\pm 11.9$ & $47.6\pm 8.0$ & $52.3\pm 9.3$ & $52.9\pm 5.5$ & $58.1\pm 4.8$ & $54.9\pm 4.5$ & $60.3\pm 4.3$ & $57.0\pm 4.6$ & $62.4\pm 2.8$ \\
\multicolumn{1}{|c|}{} & {\bf SIMT} & $\bm{39.2\pm 8.2}$ & $\bm{44.7\pm 10.9}$ & $\bm{49.5\pm 4.5}$ & $\bm{55.9\pm 4.5}$ & $\bm{52.4\pm 4.0}$ & $\bm{56.9\pm 4.6}$ & $\bm{56.2\pm 3.3}$ & $\bm{61.4\pm 3.6}$ & $\bm{57.5\pm 4.4}$ & $\bm{61.8\pm 5.2}$ & $\bm{60.8\pm 3.7}$ & $\bm{65.4\pm 3.2}$ \\ \hline
\multicolumn{1}{|c|}{\multirow{5}{*}{Pubmed}} & Greedy & $15.0\pm 4.6$ & $27.4\pm 8.6$ & $16.6\pm 7.8$ & $29.1\pm 10.6$ & $17.1\pm 8.6$ & $29.4\pm 10.9$ & $17.6\pm 9.0$ & $29.7\pm 11.1$ & $18.2\pm 9.5$ & $30.0\pm 11.1$ & $19.3\pm 9.8$ & $30.7\pm 11.2$ \\
\multicolumn{1}{|c|}{} & ASCV & $16.6\pm 4.2$ & $32.4\pm 9.5$ & $18.0\pm 4.6$ & $33.1\pm 9.4$ & $18.9\pm 4.8$ & $33.5\pm 9.2$ & $20.5\pm 6.5$ & $34.2\pm 9.9$ & $23.1\pm 8.6$ & $36.9\pm 10.1$ & $24.6\pm 9.9$ & $37.8\pm 10.9$ \\
\multicolumn{1}{|c|}{} & Greedy(P) & $18.4\pm 9.5$ & $30.0\pm 11.4$ & $21.8\pm 14.7$ & $32.2\pm 14.0$ & $23.8\pm 17.3$ & $33.6\pm 15.7$ & $24.7\pm 17.1$ & $34.2\pm 15.3$ & $25.0\pm 15.7$ & $34.4\pm 14.3$ & $26.8\pm 17.3$ & $36.0\pm 15.3$ \\
\multicolumn{1}{|c|}{} & ASCV(P) & $23.5\pm 11.4$ & $36.4\pm 12.5$ & $25.3\pm 10.9$ & $37.6\pm 12.4$ & $29.4\pm 12.3$ & $40.6\pm 13.0$ & $31.1\pm 12.2$ & $42.0\pm 12.6$ & $33.5\pm 13.8$ & $43.4\pm 13.0$ & $36.5\pm 15.6$ & $45.5\pm 14.2$ \\
\multicolumn{1}{|c|}{} & {\bf SIMT} & $\bm{39.1\pm 21.1}$ & $\bm{49.1\pm 19.6}$ & $\bm{44.4\pm 19.9}$ & $\bm{52.8\pm 18.1}$ & $\bm{48.9\pm 20.9}$ & $\bm{56.4\pm 19.1}$ & $\bm{52.6\pm 19.5}$ & $\bm{59.0\pm 18.2}$ & $\bm{55.1\pm 17.1}$ & $\bm{61.4\pm 15.5}$ & $\bm{58.1\pm 13.9}$ & $\bm{63.3\pm 13.6}$ \\ \hline
\end{tabular}%
}
\end{table*}

\begin{table*}[ht!]
\centering
\caption{Node classification performance under different budgets using GraphSage architecture for $n_i=1$}
\label{tab:GraphSage}
\resizebox{\linewidth}{!}{%
\begin{tabular}{|cc|cc|cc|cc|cc|cc|cc|}
\hline
\multicolumn{2}{|c|}{budget} & \multicolumn{2}{c|}{50} & \multicolumn{2}{c|}{100} & \multicolumn{2}{c|}{150} & \multicolumn{2}{c|}{200} & \multicolumn{2}{c|}{250} & \multicolumn{2}{c|}{300} \\ \hline
\multicolumn{2}{|c|}{metric} & MacroF1 & MicroF1 & MacroF1 & MicroF1 & MacroF1 & MicroF1 & MacroF1 & MicroF1 & MacroF1 & MicroF1 & MacroF1 & MicroF1 \\ \hline
\multicolumn{1}{|c|}{\multirow{5}{*}{Cora}} & Greedy & $16.7\pm 3.8$ & $19.8\pm 2.7$ & $22.7\pm 5.4$ & $26.1\pm 4.7$ & $26.6\pm 6.4$ & $29.7\pm 6.0$ & $30.2\pm 7.0$ & $32.7\pm 4.9$ & $31.8\pm 6.9$ & $32.8\pm 5.0$ & $34.8\pm 9.0$ & $37.6\pm 8.0$ \\
\multicolumn{1}{|c|}{} & ASCV & $15.2\pm 2.3$ & $18.1\pm 2.0$ & $23.9\pm 4.2$ & $24.2\pm 3.2$ & $29.8\pm 6.4$ & $29.4\pm 4.1$ & $31.6\pm 6.6$ & $30.1\pm 6.3$ & $37.2\pm 7.9$ & $34.7\pm 6.7$ & $40.6\pm 5.1$ & $37.9\pm 6.2$ \\
\multicolumn{1}{|c|}{} & Greedy(P) & $36.8\pm 12.4$ & $46.9\pm 12.3$ & $42.9\pm 15.4$ & $52.9\pm 14.2$ & $43.7\pm 15.8$ & $52.8\pm 14.7$ & $47.3\pm 16.2$ & $56.3\pm 14.1$ & $47.8\pm 15.0$ & $56.6\pm 11.9$ & $49.5\pm 14.5$ & $58.5\pm 11.2$ \\
\multicolumn{1}{|c|}{} & ASCV(P) & $43.9\pm 9.7$ & $51.6\pm 7.4$ & $55.1\pm 5.3$ & $60.1\pm 6.2$ & $63.2\pm 8.3$ & $66.9\pm 6.5$ & $67.7\pm 6.6$ & $70.0\pm 5.7$ & $71.6\pm 5.1$ & $73.7\pm 4.3$ & $72.0\pm 4.2$ & $73.9\pm 3.4$ \\
\multicolumn{1}{|c|}{} & {\bf SIMT} & $\bm{50.3\pm 6.4}$ & $\bm{57.6\pm 6.6}$ & $\bm{64.1\pm 7.4}$ & $\bm{67.7\pm 4.9}$ & $\bm{65.9\pm 5.1}$ & $\bm{69.1\pm 2.9}$ & $\bm{70.9\pm 3.1}$ & $\bm{73.8\pm 2.3}$ & $\bm{73.5\pm 3.6}$ & $\bm{76.2\pm 3.3}$ & $\bm{73.6\pm 3.0}$ & $\bm{76.7\pm 2.4}$ \\ \hline
\multicolumn{1}{|c|}{\multirow{5}{*}{Citeseer}} & Greedy & $15.8 \pm 3.7$ & $20.1 \pm 6.1$ & $20.1 \pm 1.7$ & $25.6 \pm 3.5$ & $24.2 \pm 3.1$ & $30.0 \pm 4.9$ & $26.1 \pm 4.1$ & $32.1 \pm 5.3$ & $29.2 \pm 4.5$ & $33.3 \pm 4.2$ & $31.4 \pm 4.8$ & $35.9 \pm 5.2$ \\
\multicolumn{1}{|c|}{} & ASCV & $16.7 \pm 3.2$ & $19.7 \pm 3.5$ & $22.4 \pm 2.6$ & $24.5 \pm 3.9$ & $25.0 \pm 2.2$ & $25.9 \pm 3.0$ & $29.6 \pm 3.5$ & $30.0 \pm 4.7$ & $32.5 \pm 2.6$ & $34.5 \pm 4.3$ & $34.9 \pm 3.9$ & $34.9 \pm 5.8$ \\
\multicolumn{1}{|c|}{} & Greedy(P) & $30.1 \pm 8.5$ & $38.5 \pm 9.3$ & $38.3 \pm 9.8$ & $46.1 \pm 9.2$ & $43.2 \pm 9.8$ & $50.8 \pm 8.8$ & $45.5 \pm 9.6$ & $53.8 \pm 9.3$ & $49.4 \pm 10.9$ & $57.7 \pm 9.4$ & $50.2 \pm 10.4$ & $58.6 \pm 8.9$ \\
\multicolumn{1}{|c|}{} & ASCV(P) & $35.2 \pm 8.8$ & $40.6 \pm 10.5$ & $46.9 \pm 7.7$ & $53.4 \pm 7.3$ & $55.2 \pm 6.4$ & $\bm{59.6 \pm 5.3}$ & $56.0 \pm 5.8$ & $60.8 \pm 3.6$ & $57.6 \pm 3.0$ & $62.4 \pm 3.1$ & $\bm{60.9 \pm 2.8}$ & $65.0 \pm 3.0$ \\
\multicolumn{1}{|c|}{} & {\bf SIMT} & $\bm{47.6 \pm 7.0}$ & $\bm{52.8 \pm 7.5}$ & $\bm{54.0 \pm 3.1}$ & $\bm{57.4 \pm 3.0}$ & $\bm{56.3 \pm 3.1}$ & $59.4 \pm 3.0$ & $\bm{58.4 \pm 2.7}$ & $\bm{62.3 \pm 2.8}$ & $\bm{59.9 \pm 1.8}$ & $\bm{63.2 \pm 2.3}$ & $60.7 \pm 3.2$ & $\bm{66.3 \pm 2.0}$ \\ \hline
\multicolumn{1}{|c|}{\multirow{5}{*}{Pubmed}} & Greedy & $25.8 \pm 4.1$ & $31.4 \pm 5.0$ & $27.7 \pm 2.7$ & $35.9 \pm 2.9$ & $26.3 \pm 2.8$ & $37.6 \pm 4.8$ & $25.9 \pm 3.9$ & $38.5 \pm 6.3$ & $27.0 \pm 4.9$ & $38.9 \pm 5.8$ & $27.5 \pm 3.4$ & $39.9 \pm 5.3$ \\
\multicolumn{1}{|c|}{} & ASCV & $15.2 \pm 2.3$ & $18.1 \pm 2.0$ & $23.9 \pm 4.2$ & $24.2 \pm 3.2$ & $29.8 \pm 6.4$ & $29.4 \pm 4.1$ & $31.6 \pm 6.6$ & $30.1 \pm 6.3$ & $37.2 \pm 7.9$ & $34.7 \pm 6.7$ & $40.6 \pm 5.1$ & $37.9 \pm 6.2$ \\
\multicolumn{1}{|c|}{} & Greedy(P) & $36.8 \pm 12.4$ & $46.9 \pm 12.3$ & $42.9 \pm 15.4$ & $52.9 \pm 14.2$ & $43.7 \pm 15.8$ & $52.8 \pm 14.7$ & $47.3 \pm 16.2$ & $56.3 \pm 14.1$ & $47.8 \pm 15.0$ & $56.6 \pm 11.9$ & $49.5 \pm 14.5$ & $58.5 \pm 11.2$ \\
\multicolumn{1}{|c|}{} & ASCV(P) & $47.5 \pm 8.6$ & $53.8 \pm 8.2$ & $46.9 \pm 12.0$ & $52.6 \pm 10.6$ & $49.8 \pm 12.9$ & $54.4 \pm 11.2$ & $50.7 \pm 14.0$ & $56.0 \pm 12.3$ & $57.4 \pm 13.7$ & $61.0 \pm 11.4$ & $57.8 \pm 14.1$ & $60.8 \pm 11.8$ \\
\multicolumn{1}{|c|}{} & {\bf SIMT} & $\bm{56.8 \pm 9.3}$ & $\bm{60.5 \pm 7.5}$ & $\bm{61.0 \pm 12.5}$ & $\bm{63.5 \pm 10.2}$ & $\bm{62.9 \pm 13.0}$ & $\bm{65.3 \pm 11.2}$ & $\bm{65.1 \pm 10.1}$ & $\bm{66.9 \pm 8.8}$ & $\bm{68.4 \pm 5.6}$ & $\bm{69.9 \pm 5.0}$ & $\bm{67.1 \pm 10.5}$ & $\bm{68.7 \pm 9.7}$ \\ \hline
\end{tabular}%
}
\end{table*}

\begin{table*}[ht!]
\centering
\caption{Node classification performance under different budgets using GAT architecture for $n_i=1$}
\label{tab:GAT}
\resizebox{\linewidth}{!}{%
\begin{tabular}{|cc|cc|cc|cc|cc|cc|cc|}
\hline
\multicolumn{2}{|c|}{budget} & \multicolumn{2}{c|}{50} & \multicolumn{2}{c|}{100} & \multicolumn{2}{c|}{150} & \multicolumn{2}{c|}{200} & \multicolumn{2}{c|}{250} & \multicolumn{2}{c|}{300} \\ \hline
\multicolumn{2}{|c|}{metric} & MacroF1 & MicroF1 & MacroF1 & MicroF1 & MacroF1 & MicroF1 & MacroF1 & MicroF1 & MacroF1 & MicroF1 & MacroF1 & MicroF1 \\ \hline
\multicolumn{1}{|c|}{\multirow{5}{*}{Cora}} & Greedy & $20.0 \pm 8.6$ & $29.2 \pm 9.6$ & $25.7 \pm 6.8$ & $34.5 \pm 9.0$ & $28.6 \pm 9.7$ & $37.1 \pm 10.3$ & $31.8 \pm 11.9$ & $40.5 \pm 11.3$ & $34.6 \pm 12.0$ & $44.6 \pm 9.4$ & $36.4 \pm 11.5$ & $45.9 \pm 9.4$ \\
\multicolumn{1}{|c|}{} & ASCV & $20.9 \pm 8.0$ & $26.0 \pm 9.0$ & $34.2 \pm 8.9$ & $38.8 \pm 7.1$ & $39.6 \pm 7.3$ & $45.0 \pm 4.9$ & $47.7 \pm 9.3$ & $52.2 \pm 5.6$ & $53.9 \pm 6.5$ & $56.5 \pm 4.4$ & $57.6 \pm 5.5$ & $59.8 \pm 4.0$ \\
\multicolumn{1}{|c|}{} & Greedy(P) & $25.1 \pm 11.7$ & $35.8 \pm 12.1$ & $30.7 \pm 13.8$ & $41.1 \pm 13.5$ & $33.6 \pm 15.0$ & $43.7 \pm 14.5$ & $36.6 \pm 18.4$ & $46.6 \pm 16.5$ & $38.1 \pm 16.4$ & $47.7 \pm 13.5$ & $39.8 \pm 16.7$ & $50.3 \pm 13.7$ \\
\multicolumn{1}{|c|}{} & ASCV(P) & $30.8 \pm 11.4$ & $36.6 \pm 9.5$ & $41.5 \pm 12.0$ & $49.6 \pm 9.9$ & $50.3 \pm 11.2$ & $57.2 \pm 8.7$ & $60.2 \pm 8.5$ & $64.8 \pm 7.5$ & $63.2 \pm 6.4$ & $67.8 \pm 5.5$ & $70.2 \pm 6.0$ & $72.9 \pm 6.1$ \\
\multicolumn{1}{|c|}{} & {\bf SIMT} & $\bm{35.2 \pm 6.8}$ & $\bm{45.9 \pm 8.5}$ & $\bm{49.2 \pm 7.8}$ & $\bm{59.4 \pm 7.1}$ & $\bm{57.3 \pm 6.2}$ & $\bm{65.7 \pm 5.0}$ & $\bm{62.7 \pm 8.2}$ & $\bm{70.4 \pm 5.5}$ & $\bm{69.2 \pm 6.6}$ & $\bm{74.3 \pm 4.5}$ & $\bm{71.7 \pm 5.4}$ & $\bm{76.1 \pm 4.3}$ \\ \hline
\multicolumn{1}{|c|}{\multirow{5}{*}{Citeseer}} & Greedy & $11.6 \pm 5.0$ & $19.8 \pm 8.3$ & $16.9 \pm 7.6$ & $25.7 \pm 10.0$ & $20.3 \pm 8.9$ & $28.2 \pm 8.9$ & $25.6 \pm 10.5$ & $33.0 \pm 10.5$ & $28.1 \pm 11.4$ & $34.6 \pm 10.3$ & $29.5 \pm 11.5$ & $37.5 \pm 9.4$ \\
\multicolumn{1}{|c|}{} & ASCV & $12.4 \pm 3.7$ & $19.5 \pm 5.3$ & $20.7 \pm 4.8$ & $26.0 \pm 7.7$ & $28.7 \pm 5.9$ & $34.7 \pm 7.2$ & $34.9 \pm 3.9$ & $38.9 \pm 5.6$ & $38.7 \pm 6.2$ & $44.3 \pm 3.5$ & $44.0 \pm 5.3$ & $49.9 \pm 3.6$ \\
\multicolumn{1}{|c|}{} & Greedy(P) & \multicolumn{1}{l}{$15.3 \pm 7.8$} & \multicolumn{1}{l|}{$25.6 \pm 9.8$} & \multicolumn{1}{l}{$20.2 \pm 9.4$} & \multicolumn{1}{l|}{$31.1 \pm 10.4$} & \multicolumn{1}{l}{$24.0 \pm 11.9$} & \multicolumn{1}{l|}{$34.1 \pm 10.6$} & \multicolumn{1}{l}{$30.7 \pm 13.4$} & \multicolumn{1}{l|}{$39.9 \pm 11.9$} & \multicolumn{1}{l}{$32.5 \pm 13.0$} & \multicolumn{1}{l|}{$41.7 \pm 11.5$} & \multicolumn{1}{l}{$36.6 \pm 14.5$} & \multicolumn{1}{l|}{$46.0 \pm 13.0$} \\
\multicolumn{1}{|c|}{} & ASCV(P) & $19.8 \pm 8.1$ & $28.3 \pm 9.8$ & $31.7 \pm 8.0$ & $38.7 \pm 9.9$ & $41.7 \pm 5.1$ & $48.1 \pm 6.0$ & $47.5 \pm 8.6$ & $54.6 \pm 8.3$ & $53.7 \pm 7.4$ & $60.7 \pm 6.2$ & $55.8 \pm 7.0$ & $62.7 \pm 5.0$ \\
\multicolumn{1}{|c|}{} & {\bf SIMT} & $\bm{31.5 \pm 8.9}$ & $\bm{36.7 \pm 9.8}$ & $\bm{46.0 \pm 7.5}$ & $\bm{51.5 \pm 8.1}$ & $\bm{49.7 \pm 4.1}$ & $\bm{56.4 \pm 3.8}$ & $\bm{53.3 \pm 4.0}$ & $\bm{60.6 \pm 4.2}$ & $\bm{58.0 \pm 3.8}$ & $\bm{63.8 \pm 4.7}$ & $\bm{59.1 \pm 2.1}$ & $\bm{66.0 \pm 1.9}$ \\ \hline
\multicolumn{1}{|c|}{\multirow{5}{*}{Pubmed}} & Greedy & \multicolumn{1}{l}{$16.2 \pm 6.2$} & \multicolumn{1}{l|}{$29.0 \pm 9.9$} & \multicolumn{1}{l}{$16.6 \pm 5.3$} & \multicolumn{1}{l|}{$28.5 \pm 8.6$} & \multicolumn{1}{l}{$21.7 \pm 10.6$} & \multicolumn{1}{l|}{$33.6 \pm 11.0$} & \multicolumn{1}{l}{$23.0 \pm 12.4$} & \multicolumn{1}{l|}{$34.9 \pm 12.3$} & \multicolumn{1}{l}{$22.3 \pm 11.1$} & \multicolumn{1}{l|}{$33.6 \pm 11.6$} & \multicolumn{1}{l}{$22.0 \pm 11.5$} & \multicolumn{1}{l|}{$33.6 \pm 11.4$} \\
\multicolumn{1}{|c|}{} & ASCV & $17.1 \pm 4.7$ & $32.6 \pm 9.6$ & $20.3 \pm 7.0$ & $35.2 \pm 9.6$ & $23.3 \pm 7.8$ & $37.5 \pm 8.7$ & $23.5 \pm 8.6$ & $37.6 \pm 8.7$ & $23.8 \pm 8.6$ & $37.7 \pm 9.0$ & $28.8 \pm 10.4$ & $42.2 \pm 8.7$ \\
\multicolumn{1}{|c|}{} & Greedy(P) & $16.9 \pm 9.1$ & $29.4 \pm 11.0$ & $15.2 \pm 3.8$ & $27.2 \pm 8.4$ & $15.3 \pm 4.8$ & $27.2 \pm 8.7$ & $17.8 \pm 9.9$ & $30.0 \pm 12.0$ & $19.1 \pm 7.1$ & $30.0 \pm 8.3$ & $21.0 \pm 10.3$ & $32.3 \pm 10.2$ \\
\multicolumn{1}{|c|}{} & ASCV(P) & $17.2 \pm 4.9$ & $32.7 \pm 9.7$ & $17.8 \pm 5.5$ & $32.9 \pm 9.9$ & $21.5 \pm 9.7$ & $36.1 \pm 10.4$ & $24.7 \pm 11.6$ & $38.2 \pm 12.1$ & $25.6 \pm 11.4$ & $38.7 \pm 12.3$ & $29.6 \pm 11.3$ & $41.3 \pm 10.5$ \\
\multicolumn{1}{|c|}{} & {\bf SIMT} & $\bm{19.5 \pm 5.5}$ & $\bm{34.4 \pm 9.3}$ & $\bm{30.0 \pm 12.1}$ & $\bm{43.8 \pm 13.0}$ & $\bm{32.5 \pm 14.5}$ & $\bm{43.9 \pm 14.7}$ & $\bm{32.0 \pm 13.3}$ & $\bm{42.9 \pm 12.0}$ & $\bm{41.9 \pm 14.0}$ & $\bm{52.0 \pm 12.9}$ & $\bm{44.5 \pm 20.4}$ & $\bm{54.9 \pm 19.3}$ \\ \hline
\end{tabular}%
}
\end{table*}

\section{Comparison of GNN architectures}
\label{app:models}

To test the effect of GNN architectures, we also conduct experiments on scenario with $n_i=1, i\in O$ using three representative GNN architectures: GIN \cite{xu2018powerful}, GraphSage \cite{hamilton2017inductive}, and GAT \cite{velivckovic2017graph}. These architectures have emerged as alternatives to the traditional GCN \cite{kipf2016semi} and have gained widespread acceptance within the GNN community. The configuration of these architectures are the same as that of GCN as shown in GNN models part in  Sec.~\ref{sec:setup}. 
To eliminate the randomness induced by subgraphs, we conduct experiments on the case where all data owners are individuals.

The experiment results are shown in Tables~\ref{apptab:quantity}, \ref{apptab:GCN}, \ref{tab:GIN}, \ref{tab:GraphSage}, and~\ref{tab:GAT}.
In general, SIMT outperforms its benchmarks in all experiments. 
When compared to the Greedy mechanism, SIMT exhibits an improvement ranging from $10\%$ to $ 40\%$ in both MacroF1 and MicroF1, consistently confirming the effectiveness of SIMT.
However, a minor fluctuation is observed when compared to ASCV(P) within the GraphSage architecture on the Citeseer dataset.
More specifically, when the budget is 1500, SIMT demonstrates a $1.1\%$ improvement in MacroF1, albeit a $0.2\%$ decrease in MicroF1. Furthermore, when the budget is increased to 3000, SIMT shows a slight decrease by $0.2\%$ in MacroF1, but a notable increase by $1.3\%$ in MicroF1. Consequently, this fluctuation does not detract from the overall effectiveness of SIMT.

When comparing the effects of different architectures, it is observed that the GraphSage architecture generally outperforms the others. This superior performance is particularly noticeable when the budget is low. Additionally, the GraphSage architecture exhibits lower variance, which implies robust performance. This could potentially be attributed to the fact that GraphSage does not heavily depend on node features and is well suited to spectral features \cite{velivckovic2017graph}. This compatibility allows it to effectively leverage the spectral properties of the data, potentially contributing to its robust performance across various experiments. 
Furthermore, when the budget is set to 3000, the GraphSage architecture underperforms the GIN architecture on the Cora dataset. This could potentially be attributed to the theoretically stronger learnability of GIN compared to GraphSage \cite{xu2018powerful}.
Last, the GAT architecture performs the worst among the four architectures. This could be due to the fact that GAT employs an attention mechanism, which has significantly more parameters and thus requires a sufficient amount of feature and label data for parameter training.

\section{Ablation study for $n_i=1$} \label{appsec:ablation}
We conduct ablation studies across the three datasets and present the average test accuracy for $n_i=1$. 
As shown in Table~\ref{apptab:ablationstudy}, {\bf (1)} the accuracy of SIMT without clustering diminishes by $8.3\%$ in Cora, $3.7\%$ in Citeseer, and $13.7\%$ in Pubmed. This significant reduction underscores the crucial role of structural clustering. Without clustering, there is a high probability that the procured data are unevenly distributed across the classes, leading to a biased training dataset.
{\bf (2)} The accuracy of SIMT without structural informativeness drops by $1.8\%$ in Cora, $2.1\%$ in Citeseer, and $2.9\%$ in Pubmed. {\bf (3)} The accuracy reduction for SIMT without structural representativeness is $1.8\%$ in Cora, $1.5\%$ in Citeseer, and $2.2\%$ in Pubmed. Overall, each component distinctly enhances SIMT's performance.

In the scenario with $n_i=1, i\in O$, the results show similar trends as the scenario with $n_i\geq 1, i\in O$: SIMT performs the best under all budgets, followed by ASCV/ASCV(P), and Greedy/Greedy(P) performs the worst. It is worth to notice that SIMT performs even better when $n_i=1, i\in O$ than when $n_i\geq 1, i\in O$. In particular, SIMT improves $20\% \sim 40\%$ in both MacroF1 and MicroF1 in this case; See more details in Table~\ref{apptab:GCN} of App.~\ref{app:models}. This can be attributed to that when $n_i=1, i\in O$, fewer edges are missing within subgraphs, allowing the structural importance scores calculated by our mechanism to more accurately reflect the true structural importance of each node.

\section{Limitations and future work}
Even though our SIMT mechanism shows the incentive properties theoretically and the accuracy empirically, we acknowledge certain limitations that offer avenues for future research. 
\begin{itemize}
    \item This paper addresses issues related to pre-purchase inspections. However, there are concerns about privacy breaches, as the raw data shared with data brokers and the trained models returned to data consumers could expose private information of the data owners. Future work could enhance the SIMT mechanism by incorporating considerations of these privacy issues
    \item  In the design of the SIMT mechanism, we assess the structural importance once and subsequently make a single purchasing decision for all data owners. An alternative approach could involve an adaptive purchasing method, where data points are bought in several times, each time reassessing the structural importance in light of newly acquired data. The primary rationale for the existing design is efficiency compared to this adaptive method. Nonetheless, developing an efficient adaptive model training mechanism presents an exciting avenue for future research.
    \item The efficacy of the SIMT mechanism heavily depends on the quality of clustering; more accurate structural importance assessment is achieved when clusters are closer to the true classes. Our SIMT mechanism employs a classical VGAE model \cite{kipf2016variational} as the $\mathsf{Clustering}$. However, our findings indicate that this approach is not always robust. Exploring a broader range of GNN models for clustering might yield more stable results and enhance the performance of the SIMT mechanism, although such investigations fall outside the scope of this paper
\end{itemize}


\end{document}